\documentclass[final,authoryear,5p,times,twocolumn]{elsarticle}
\usepackage{amsmath}
\usepackage{amssymb} 

\usepackage{tikz}
\usetikzlibrary{arrows,matrix,positioning,patterns}
\usetikzlibrary{calc}
\usetikzlibrary{external}
\usepackage[utf8]{inputenc}

\usepackage{pgfplots}
\usepackage{pgfplotstable}
\usepgfplotslibrary{groupplots}
\usetikzlibrary{plotmarks}

\newdefinition{rmk}{Remark}
\newproof{pf}{Proof}
\newproof{pot}{Proof of Theorem \ref{thm2}}

\usepackage{algorithm} 
\usepackage{algpseudocode}

\usepackage{amsthm}

\newtheorem{assumption}{Assumption}
\newtheorem{proposition}{Proposition}

\newtheorem{remark}{Remark}

\newtheorem{problem}{Problem}

\usepackage{amsmath}
\usepackage{amssymb} 

\usepackage{eurosym}
\usepackage{colortbl}

\usepackage[labelformat=simple]{subcaption}
\DeclareCaptionLabelSeparator{periodspace}{.\quad}
\usepackage{color,framed}
\usepackage[utf8]{inputenc} 
\usepackage{tabularx}

\usepackage{rotating}
\usepackage{graphicx}
\usepackage{lscape}
\allowdisplaybreaks 
\usepackage{longtable}

\makeatletter
\def\ps@pprintTitle{%
 \let\@oddhead\@empty
 \let\@evenhead\@empty
 \def\@oddfoot{}%
 \let\@evenfoot\@oddfoot}
\makeatother

\begin{document}

\begin{frontmatter}
\title{Coupling of Crop Assignment and Vehicle Routing for Harvest Planning in Agriculture}

\author{Mogens Graf Plessen\corref{cor1}}
%
%
\cortext[cor1]{\texttt{mgplessen@gmail.com}}

\begin{abstract}
A method for harvest planning based on the coupling of crop assignment with vehicle routing is presented. Given a setting with multiple fields, a path network connecting these, multiple depots at which a number of harvesters are initially located, the main question addressed is: Which crop out of a set of different crops to assign to each field when accounting for the given setting? It must be answered by every farm manager at the beginning of every work-cycle starting with plant seeding and ending with harvesting. Rather than solving an assignment problem only, it is here also accounted for the connectivity between fields. In practice, fields are often located distant apart. Traveling costs of machinery and limited harvesting windows demand optimised operation and route planning. Therefore, the proposed method outputs crop assignments to fields and simultaneously determines crop-tours, i.e., optimised sequences in which to service fields of the same crop during harvest. It is of particular relevance for larger farms and groups of farms that collaborate and share machinery. Integer programming based exact algorithms are derived. For large numbers of fields, where these algorithms may not be tractable due to computational constraints, elements of clustering and the solution of local Traveling Salesman Problems are added, thereby on the one hand rendering the method heuristic and in general suboptimal, but on the other hand maintaining large-scale applicability. 
\end{abstract}
\end{frontmatter}


\begin{table}
\centering
\begin{tabular}{|ll|}
\hline
\multicolumn{2}{|l|}{MAIN NOMENCLATURE}\\
$\mathcal{D}$ & Set of $D$ depots, indexed by $d,i,j\in\mathcal{D}$. \\
$\mathcal{L}$ & Set of $L$ fields, indexed by $l,i,j\in\mathcal{L}$. \\
$\mathcal{K}$ & Set of $K$ crops, indexed by $k\in\mathcal{K}$. \\
$x_{ij}^k$ & Decision variable for edge $(i,j)$ and crop $k$.\\
$\delta_{l}^k$ & Decision variable for vertex $l$ and crop $k$.\\ 
$\xi^d$ & Decision variable for depot $d$.\\
$\gamma$ & Decision variable for the nr. of active crops.\\
$\tilde{c}_{ij}^k$ & Cost per harvester for arc $(i,j)$ and crop $k$.\\
$c_{ij}^k$ & Cost for arc $(i,j)$ and crop $k$.\\
$r_{l}^k$ & Monetary return for crop $k$ on field $l$.\\
$z^d$ & Cost coefficient for active depot $d$.\\
$m$ & Cost coefficient per active crop.\\
$N_d^{\text{harv},k}$ & Nr. of harvesters at depot $d$ for crop-tour $k$.\\
$\tilde{k}$ & Parameter for $\tilde{k}$-means clustering.\\
$J$ & Monetary return (revenue minus costs).\\
\multicolumn{2}{|l|}{ABBREVIATIONS}\\
IP/LP & Integer/Linear Programming\\
TSP/mTSP & Traveling Salesman Problem/multiple TSP\\
\hline
\end{tabular}
\end{table}

\section{Introduction\label{sec_intro}}

Agriculture is a diverse field ranging from biotech to autonomous robots and finance. It is also closely related to logistics. According to \cite{ahumada2009application}, there are four main functional areas for the agri-food supply chain: production, harvesting, storage and distribution. This paper focuses on model-based production planning. In fact, in view of recent plunges of agricultural commodity prices according to the \cite{FTlowwheatprice} that threaten the sustainability of not few farmers, efficiency improvements in production are essential. The decision on the assignment of crops to fields is crucial in that it determines the complete work-cycle. In common practice today, fields are often first manually clustered according to geographical locations before a crop is assigned uniformly to all fields of each cluster, whereby it is often accounted for crop rotation (\cite{havlin1990crop}) in order to reduce soil erosion and to increase soil fertility. The spatial clustering is done for faster harvesting. A trend among farmers in Europe is to collaborate in form of limited companies for the sharing of machinery. Not seldomly conflicts arise about the sequence in which to harvest multiple fields of identical crops but various owners. This paper is motivated by providing a mathematical modeling framework for crop assignment to fields when also accounting for the path network connecting these fields and depots of harvesters.

The basic multiple Traveling Salesman Problem (mTSP) describes the objective of finding total tour cost-minimizing routes for multiple salesmen that all start and end at a single depot, whereby all vertices are visited once by exactly one salesman, see \cite{bektas2006multiple}. Nonnegative edge cost can refer to, e.g., monetary, space or time units. When accounting for various demands at each vertex and limiting the capacity of vehicles (salesmen), the problem is referred to as the capacitated Vehicle Routing Problem (VRP). Variations include the VRP with time windows, with backhauls and with pickup and delivery, see \cite{toth2014vehicle}. The applications are manifold. For example, for vehicle routing with real-time information see~\cite{kim2005optimal} and the references therein. Recently, there has been increased interest in applying logistical optimisation in agriculture for scheduling, routing and fleet management, see \cite{basnet2006scheduling}, \cite{ferrer2008optimization}, \cite{marques2014tactical}, \cite{conesa2016route} and \cite{sorensen2010conceptual}. Special focus was on the coordination of machinery teams distinguishing between primary (harvester) and service (transport) units referred to as PUs and SUs, see \cite{bochtis2009vehicle}, \cite{jensen2012field}, \cite{seyyedhasani2017using} and \cite{orfanou2013scheduling}. All of these references assume that fields with assigned crops are given. To the best of the author's knowledge, the optimised assignment of crops to fields and \emph{simultaneously} accounting for vehicle routing and other constraints for optimised harvest planning has not been discussed in the literature before. Such strategic assignment must be conducted once per year and at the \emph{beginning} of the (yearly) work-cycle, thereby decisively affecting the complete agricultural production-cycle, as the first step within a two-layered framework. The second layer involves coordinations of PUs and SUs, exploiting all of aforementioned references, and is to be conducted at the \emph{end} of the work-cylce during harvest.

The contribution of this paper is a novel method that can (a) assist farm managers in the planning of crop assignments to available fields and (b) simultaneously determines crop-tours, i.e., optimised sequences in which to service fields of the same crop during harvest. Eight different integer programs (IPs) are formulated corresponding to different problem setups.  

The remaining paper is organised as follows. Problem and notation are formulated in Sect. \ref{sec_problFormulation}. Algorithms are derived in Sect. \ref{sec_ProblApproach}. Extensions are discussed in Sect. \ref{sec_Extensions}. Numerical examples are given in Sect. \ref{sec_numerical_experiments}, before concluding with Sect. \ref{sec_conclusion}.

\section{Problem Formulation and Notation\label{sec_problFormulation}}

The general problem is formulated in literal form. Then, mathematical notation is introduced based on which the preferred eight IPs and the main algorithm are derived and discussed in Sect. \ref{sec_ProblApproach}.

\subsection{Problem Formulation\label{subsec_ProblDef}}

\begin{figure}[!ht]
\vspace{0.3cm}
\centering
\includegraphics[width=8cm]{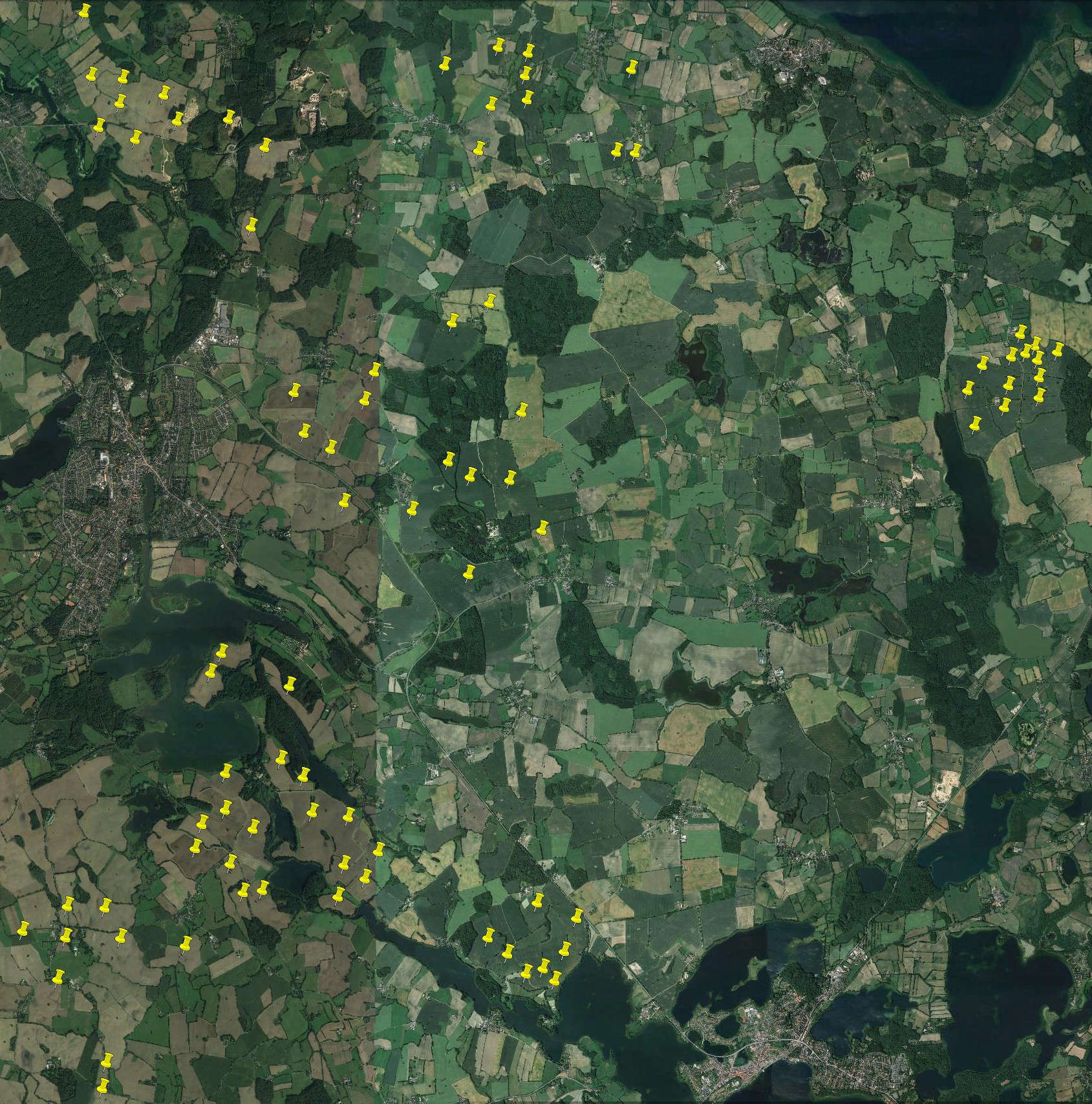}
\caption{Problem visualisation by example. Yellow markers indicate fields to be served by 3 collaborating farms. Overall, there are 85 fields. The satellite image shows an area of 15.9 $\times$ 16.3km. The path network connecting the fields is curvy and often along rural gravel roads only permitting slow traveling speeds. The overall field coverage area is more than 1700 hectares. Traveling distances between pairs of fields is between meters up to dozens of km.}
\label{fig:realWorldScenario}
\end{figure}

\begin{problem}\label{probl_def}
Suppose a setting in which multiple fields, a path network connecting these fields, multiple depots, and potentially multiple harvesters located initially at each depot are given. Then, at the beginning of every work-cycle (beginning with planting and ending with harvesting), a crop has to be assigned to all available fields. This entails the following questions:
\begin{enumerate}
\item Which \emph{crop} to optimally assign to each field?
\item In what \emph{sequence} to optimally service all fields during harvest?
\item How to optimally dispatch \emph{multiple harvesters}, that initially may be located at multiple depots, to the multiple fields? 
\item Which fields should be serviced, which \emph{leased}, and at what prices? 
\end{enumerate}
\end{problem}
The first question decides the complete work-cycle. For optimised harvest planning, its answer must simultaneously account for questions 2) to 4). See also Fig. \ref{fig:realWorldScenario} for problem visualisation and illustration of real-world relevance due to significant inter-field distances. It is stressed that at the beginning of every work-cycle of a farm, planning decisions according to Problem \ref{probl_def} must be taken. Therefore, in this paper, optimisation problems subject to constraints are derived that permit to input data such as, e.g., cost coefficients or revenues per field and crop. Note that at the end of every work-cycle, i.e., at harvest, deviations from initial modeling have occurred. For example, the actual amount of crop harvested per field is different from the predicted one, and weather is influencing potential harvesting-windows. Thus, at the end of the work-cycle, the aforementioned second framework-layer becomes relevant, involving the coordination of PUs and SUs. However, in this paper the focus is exclusively on the first framework-layer.

\subsection{Notation\label{subsec_notation}}

Notation is mainly adopted from \cite{toth2014vehicle}. A complete graph is denoted by $G=(\mathcal{V},\mathcal{E})$, where $\mathcal{V}=\{0,\dots,D-1,D,\dots,D+L-1\}$ and $\mathcal{E}$ are vertex and edge set, respectively. Throughout, an \emph{arc} is referred to as a directed edge. The cardinality of a set of vertices is denoted by $|\cdot|$. Vertices $i\in\mathcal{D}=\{0,\dots,D-1\}$ and $i\in \mathcal{L} = \{D,\dots,D+L-1\}$ correspond to $D$ depots and $L$ fields. The $K$ different crops are indexed by $\mathcal{K}=\{0,\dots,K-1\}$. Let the number of harvesters located at a depot and suitable for a specific crop be denoted by $N_d^{\text{harv},k},~\forall d\in\mathcal{D},~\forall k\in\mathcal{K}$. Let the normalised nonnegative traveling cost per harvester and crop $k$ between a depot $d$ and a field $j$, or between two fields $i$ and $j$, be denoted by $\tilde{c}_{dj}^k$ and $\tilde{c}_{ij}^k$, respectively. Then, abbreviating $N^{\text{harv},k}=  \sum_{d\in\mathcal{D}} N_d^{\text{harv},k} $, edge costs are defined as:
\begin{subequations}
\begin{align}
c_{ij}^k &= N^{\text{harv},k} \tilde{c}_{ij}^k,~\forall i,j\in\mathcal{L},~\forall k\in\mathcal{K},\label{eq_def_cijk}\\
c_{dj}^k &= N^{\text{harv},k} \tilde{c}_{dj}^k,~\forall d\in\mathcal{D},~\forall j\in\mathcal{L},~\forall k\in\mathcal{K},\label{eq_def_cdjk}\\
c_{dj}^{k,k_\text{min}} &= \sum_{\tilde{d}\in\mathcal{D}} N_{\tilde{d}}^{\text{harv},k} \tilde{c}_{\tilde{d}j}^k,~\forall j\in\mathcal{L},~\forall k\in\mathcal{K},\label{eq_def_cdjkmin}
\end{align}  
\label{eq_def_c}
\end{subequations}
where graph $G$ is in general based on $c_{ij}^k$ and $c_{dj}^k$. Note that $c_{ij}^k$ and $c_{dj}^k$ are here defined as uniformly scaled by $N^{\text{harv},k}$, and that, in practice, $\tilde{c}_{ij}^k$ and $\tilde{c}_{dj}^k$ are proportional to inter-field path lengths. As will be shown, $c_{dj}^{k,k_\text{min}}$ are used for specific IP formulations. Similarly to \eqref{eq_def_cdjk} and \eqref{eq_def_cdjkmin} $c_{jd}^k$ and $c_{jd}^{k,k_\text{max}}$ are defined. In general, traveling costs along the same geographical paths may vary for different $k$ due to different crop-dependent harvesting machinery. Suitably, they may be modeled as varying by a constant offset. The expected revenue from growing and marketing of crop $k\in\mathcal{K}$ on field $l\in\mathcal{L}$ is denoted by $r_l^k$. Maintenance cost per depot are given by $z^d,~\forall d\in\mathcal{D}$. We assume a fixed cost of $m$ incurred for every crop.  All costs shall be in monetary units, here Euros ($\text{\euro}$). Furthermore, throughout it is assumed that arc/edge costs satisfy the triangle inequality (see \cite{fleming2013effects}).

Decision variables are discussed. It is distinguished between two major classes: natural and auxiliary decision variables. The first class comprises binary $x_{ij}^k\in\{0,1\},\forall i,j\in\mathcal{V},\forall k \in\mathcal{K}$, with $x_{ij}^k=1$ indicating arc $(i,j)$ to be element of the optimal route for crop $k$. For the symmetric case with $x_{ji}^k=x_{ij}^k$, (a) decision variable $x_{ji}^k$ can be dismissed, and (b)  $x_{dj}^k\in\{0,1,2\},\forall d\in\mathcal{D},\forall j\in\mathcal{L},~\forall k\in\mathcal{K}$ is used, which permits to indicate a visit of only field $j$ for a route corresponding to crop $k$. Furthermore, binary $\delta_l^k\in\{0,1\},\forall l\in\mathcal{L},\forall k \in \mathcal{K}$, with $\delta_l^k=1$ indicating that crop $k$ is assigned to field $l$. Integer $\gamma$ is such that $1\leq \gamma \leq K$ and  indicates the number of active crops in the optimal solution. As will be shown, auxiliary decision variables result from incorporating various \emph{logical constraints} into the IP formulations.

\section{Problem Solution\label{sec_ProblApproach}}

\subsection{Assumptions and Motivation of Harvester Group Travel\label{subsec_preliminaries}}

The derivation of the proposed eight IPs in Sect. \ref{subsec_8IPs} is based on the following assumptions and discussion.  

\begin{assumption}
Throughout the harvest of any crop, harvesters are refueled and maintained on the corresponding fields growing that crop. Thus, there is no return to depots prior to the complete coverage of all of these fields.
\label{ass_MaintenFields}
\end{assumption}

\begin{assumption}
A fixed number of harvesters is assigned for the harvest of every crop. Then, during the coverage of all fields associated with a crop, harvesters travel as a group. Thus, it is assumed that there is no dispatch of individual harvesters to individual fields.
\label{ass_HarvesterGroupTravel}
\end{assumption}

Assumptions \ref{ass_MaintenFields}, \ref{ass_HarvesterGroupTravel} and cost coefficients \eqref{eq_def_c} are defined for practical considerations discussed below. Thus, in the most general sense of Problem \ref{probl_def} they are limiting. However, on the other hand, they permit to approach Problem \ref{probl_def} based on planning routes for each crop (in the following referred to as  \emph{crop-tours}) similar to the mTSP-framework (\cite{bektas2006multiple}). Thus, a route for each crop and the fields correspond to a traveling salesman route and cities to be visited, respectively. Note, however, despite this basic analogy, the mTSP-framework is insufficient to address Problem \ref{probl_def}. In particular, crop assignment, multiple depots and additional constraints are not addressed. Therefore, eight customised IPs are derived in Sect. \ref{subsec_8IPs}. First, however, it is further elaborated on Assumption  \ref{ass_HarvesterGroupTravel}.

\begin{proposition}\label{rmk_HarvestersAsGroupOneDepot}
Suppose that multiple harvesters are initially located at \emph{one} depot to which they must return after processing all fields associated with a crop. Suppose further a graph with cost coefficients according to \eqref{eq_def_cijk} and \eqref{eq_def_cdjk}. Then, an optimal policy is that all harvesters cover the fields as a group, i.e., without distributing harvesters to different fields of the same crop.
\end{proposition}
   
\begin{proof}
The proof is by construction. For the asymmetric case with different fields ripening at different times, a unique optimal working sequence is already implied. For the symmetric case, a cost-minimizing route exists and includes exactly two edges incident to the depot vertex. This is due to nonnegativity of traveling costs and the fact that cost coefficients \eqref{eq_def_cijk} and \eqref{eq_def_cdjk} are invariant to the number of harvesters traveling along inter-field paths (uniform scaling by $N^{\text{harv},k}$). Any other initial distribution of harvesters to fields not connected to the depot vertex along the two aforementioned edges is thus suboptimal and harvester group travel is thus an optimal policy.         
\end{proof}

Several comments are made. First, accumulated inter-field path length minimisation (invariant to the number of harvesters traveling along the path) is of interest for minimisation of total non-harvesting time. By nonnegativity of edge path lengths there exists a shortest path crop-tour. Since harvesters can always work in parallel on fields since they are not constrained by each other, this shortest path should be followed by all harvesters. 

Second, harvester group travel bears more practical advantages. In general, SUs must ideally be operated such that PUs (harvesters) can operate continuously, such that any waiting times due to absent SUs for unloading are avoided. In general, the rate at which harvesters are filled is not easily predictable due to varying crop returns even within one field. Therefore, the concentration of all SUs to one field is beneficial for robustness in the sense that multiple harvesters can be served (instead of specific SU-PU couples), according to short-term freed capacities. An additional advantage is the facilitated supervision by the farm-manager.

Third, Proposition \ref{rmk_HarvestersAsGroupOneDepot} assumes a single depot. In order to differentiate between the cases that (a) all available harvesters are initially located at one specific depot, and that (b) all available harvesters are initially distributed among multiple depots, $c_{dj}^{k,k_\text{min}}$ is defined according to \eqref{eq_def_cdjkmin} in contrast to $c_{dj}^{k}$ in \eqref{eq_def_cdjk}. Then, in the IP formulations of Sect. \ref{subsec_8IPs} harvester group travel is assumed \emph{from the first field on}, and consequently using cost coefficients according to \eqref{eq_def_cijk}. Note that such harvester group travel starting from the first field of a crop-tour is practical. This is since a timely agreement upon harvest-start, e.g., a day ahead,  permits that all harvesters (from different depots) plan their travel in time and consequently start field and route coverage coordinatedly and together.

Fourth, it is remarked that if different crops have different non-overlapping harvesting times, then the same harvesters can in principle be employed sequentially for the different crop-tours. In this paper, this scenario is assumed and we order crops in $\mathcal{K}$ such that a lower index indicates an earlier harvesting time. For an application example, consider the crops in order barley, rapeseed and wheat. The alternative scenario is that harvesting times are overlapping or different crops require entirely different harvester machinery. This latter scenario is left for future work and can be approached by partitioning of harvester groups, for example, weighted (a) according to overall crop-area, or (b) according to the predicted total crop-harversing time.

Finally, note that throughout Sect. \ref{sec_ProblApproach} for the subsequent derivation of IPs it is assumed that \emph{all} fields must be served. The relaxation of this assumption is treated in Sect. \ref{sec_Extensions} when discussing financial considerations regarding the leasing of subsets of fields.

\subsection{Eight Integer Linear Programs\label{subsec_8IPs}}

Eight different IPs are formulated for the eight different problem setups  considered to be most relevant for harvest planning. First, these problem are stated only literally. Relevant constraints are summarised compactly, whereby it is distinguished between constraints used in final numerical simulations and additional constraints. Then, for brevity only the two most general IPs are stated mathematically. Nevertheless, in the final numerical experiments of Sect. \ref{sec_numerical_experiments} all eight IPs are evaluated.  

\subsubsection{Literal Formulations}

Eight integer linear programs, \textsf{IP}-1, \dots, \textsf{IP}-8, are summarised below.

\begin{itemize}

\item[\textsf{IP}-1] There is a single depot, $D=1$, from which all harvesters start and to which all harvesters return after each crop-tour. Any subset of $K$ crops can be used for planting.

\item[\textsf{IP}-2] There are multiple depots $D\geq 1$. After every crop-tour, harvesters must return to the depot from which they started. Any subset of $K$ crops can be used for planting. 

\item[\textsf{IP}-3] Among $D>1$ potential depots, the best depot w.r.t. a cost criterion is selected. Then, all available harvesters are assigned to this best depot. Note that this problem could also be addressed by separately solving \textsf{IP}-1 for each of the $D$ depots and then selecting the best solution. However, here it is solved by a single IP-formulation. Any subset of $K$ crops can be used for planting. 

\item[\textsf{IP}-4] Multiple harvesters are initially located at multiple depots $D\geq 1$. These harvesters assemble at the first field of the first crop-tour. Then, all these harvesters travel as a group for the remainder of the first and all remaining crop-tours. Then, only after coverage of the last field of the last crop-tour these harvesters return to their initial depots. For $D=1$, \textsf{IP}-4 is identical to \textsf{IP}-2. For $D>1$, for all crop-tours except the first crop-tour \textsf{IP}-4 requires to select an optimal depot similarly to \textsf{IP}-3. Any subset of $K$ crops can be used for planting.
 
\item[\textsf{IP}-$m$] Like \textsf{IP}-$n$ for all $n=1, \dots,4$, but \emph{all} of the $K$ crops must now be used: \textsf{IP}-$m$ for  $m=5, \dots,8$. 

\end{itemize}

\subsubsection{Constraints and Modeling}

In order to realise the literal problem formulations above, objective function and constraints are compactly summarised. 

\begin{enumerate}

\item The IP \emph{cost function} may include accumulated edge costs $c_{ij}^k x_{ij}^k$, negative profits $-r_l^k \delta_l^k$, depot maintenance costs $z^d \xi^d$ and crop costs $m \gamma$, whereby $c_{ij}^k$, $r_l^k$, $z^d$ and $m$ denote the predicted data  which must be assumed at the time of harvest planning.

\item \emph{Degree equations for fields} ensure that fields are visited exactly once per crop-tour. For example, for \textsf{IP}-3 with symmetric edge costs the degree equations read: $\sum_{d\in\mathcal{D}} x_{dl}^k + \sum_{i<l} x_{il}^k + \sum_{l<j} x_{lj}^k = 2 \delta_{l}^k,~\forall  l\in\mathcal{L},~\forall k\in\mathcal{K}$. For an asymmetric formulation, degree equations are split into two, differentiating an in-degree (sum of arcs entering) and out-degree (sum of arcs leaving) equation, both equal to $\delta_{l}^k,~\forall  l\in\mathcal{L},~\forall k\in\mathcal{K}$. 

\item \emph{Degree equations for depots} are conceptually identical to degree equations for fields. However, the cardinality of depot-vertices is not necessarily two. Furthermore, the resulting degree equations may be nonlinear in the original optimisation variables. Then, in order to formulate linear IPs, nonlinear degree equations can be rendered linear by the introduction of (a) auxiliary variables, and (b) the application of \emph{logical constraints} (specified below) which introduce additional linear inequality constraints. For example, for \textsf{IP}-3 the nonlinear degree equations are: $\sum_{k\in\mathcal{K}} \sum_{j\in\mathcal{L}} x_{dj}^k = 2\gamma \xi^d,~\forall d\in\mathcal{D}$, with $x_{dj}^k\in\{0,1,2\}$, $\xi^d\in\{0,1\}$, $1\leq \gamma\leq K$ and $\sum_{d\in\mathcal{D}}\xi^d = 1$. 

\item \emph{Uniqueness of crop-assignments to fields} is guaranteed through constraints $\sum_{k\in\mathcal{K}} \delta_l^k=1,~\forall l\in\mathcal{L}$.
 
\item Three classes of \emph{logical constraints} are of particular interest. They can be translated into integer linear inequalities according to \cite{williams2013model}. Because of their importance, three logical constraints are here repeated. Let $\epsilon>0$ be a small number (e.g., the machine precision), $b,b_1,b_2,b_3\in\{0,1\}$, $y\in\mathbb{R}$, and $f(x)$ such that $f:\mathbb{R}^{n_x}\rightarrow \mathbb{R}$ is linear, $n_x$ the variable dimension, $f^\text{max}=\max_{x\in\mathcal{X}} f(x)$ and $f^\text{min} = \min_{x\in\mathcal{X}} f(x)$, where $\mathcal{X}$ is a given bounded set.
\begin{enumerate}
\item  The statement ``$b=1$ if and only if $f(x)\leq 0$ and $b=0$ otherwise'' is equivalent to 
\begin{equation}
f(x)\leq f^\text{max}(1-b),\quad f(x) \geq \epsilon + (f^\text{min}-\epsilon) b. \label{eq_logic1}
\end{equation} 
\item The statement ``$b_3=1$ if and only if $b_1=1$ and $b_2=1$, and $b_3=0$ otherwise'' is equivalent to~$b_3=b_1b_2$ and is equivalent to 
\begin{equation}
b_1 + b_2 -b_3 \leq 1,\quad b_3 \leq b_1, \quad b_3 \leq b_2.\label{eq_logic2} 
\end{equation}

\item The statement ``$y=f(x)$ if $b=1$ and $y=0$ otherwise'' is equivalent to $y=b f(x)$ and is equivalent to 
\begin{subequations}
\begin{align}
& y\geq f^\text{min} b,\quad  y\leq f(x) - f^\text{min}(1-b),\\
& y\leq f^\text{max} b,\quad y\geq f(x)-f^\text{max}(1-b).
\end{align}
\label{eq_logic3}
\end{subequations}
\end{enumerate}

\item Under the assumption of a symmetric formulation (with undirected edges instead of directed arcs) and dropping crop-index $k$ for generality, the \emph{subtour elimination constraints} (SECs) according to \cite{laporte1992traveling} are given by $\sum_{i<j;i,j\in S} x_{ij} \leq |S|-1,~3\leq |S|\leq N-3,~\forall S\subseteq \mathcal{V}\backslash\{0,N-1\}$. Note that there is an exponential number of SECs.

\end{enumerate}

Additional constraints are summarised as follows.

\begin{enumerate}

\item \emph{Crop rotation constraints} (see \cite{havlin1990crop}), and constraints similarly related to soil considerations where specific soils only admit specific crops can be formulated as equality constraints, $\delta_l^k=0$, for prohibited combinations of specific field $l$ and crop $k$. 

\item \emph{Diversification constraints} read $\sum_{l\in\mathcal{L}} g_l^k \delta_l^k \leq G^k,~\forall k=0,\dots,K-2$, with $g_l^k\geq 0$ denoting weights (for example the hectares-coverage or required production means for field $l$ and crop $k$) and $G^k\geq 0$ the corresponding crop-related bounds. Crop indexed by $k=K-1$ is left unconstrained for feasibility. In general, when combining both hard and inequality constraints without additional precaution, feasibility of the resulting IP cannot be guaranteed. Infeasibility results if these constraints enforce $\sum_{k\in\mathcal{K}} \delta_l^k=0$.

\item \emph{Time constraints} can be formulated as $\sum_{d\in\mathcal{D}}\sum_{j\in\mathcal{L}} h_{dj}^k x_{dj}^k + \sum_{i<j} h_{ij}^k x_{ij}^k \leq T_\text{win}^k - \sum_{l\in\mathcal{L}} T_l^{\text{harv},k} \delta_l^k,~k\in\mathcal{K}$, where, for generality, the multi-depot case is assumed, and where $h_{dj}^k$ and $h_{ij}^k$ may denote travel time along corresponding edges, $T_\text{win}^k$ the length of the harvesting window for each crop $k$ (e.g., multiple days), and $T_l^{\text{harv},k}$ the required harvesting time (e.g., inversely proportional to the number of used harvesters) per field $l$ and crop $k$. In practice, optimal harvesting time windows may be very short due to weather constraints.

\item To account for a priori experience about different sequences in ripeness of fields, \emph{priority constraints} can be formulated. For example, relating to uncertainties, the sequence in which fields of the same crop ripe may vary, e.g., due to hillsides and varying soil. W.l.o.g., consider a statement such as ``if fields $a$ and $b$ are among the ones assigned to crop $k$, then the corresponding sequence for harvest shall be in order such that $a$ is harvested immediately after $b$''. This can be modeled as nonlinear constraint $x_{ba}^k=\delta_b^k \delta_a^k$, and can therefore be translated to linear inequalities by means of \eqref{eq_logic2}. Note that an asymmetric formulation has to be employed for all connections between vertices for which priorities are  defined. For above example, we therefore require, e.g., $x_{ba}^k\neq x_{ab}^k$. 

\end{enumerate}

\subsubsection{Two Concrete IP formulations\label{subsubsec_2concreteIPformulations}}

The two most relevant IP formulations, \textsf{IP}-3 and \textsf{IP}-4, are stated. These are most relevant since they can be simplified to all six remaining IPs. For \textsf{IP}-1, $D=1$ is enforced in \textsf{IP}-3. For \textsf{IP}-2 cost coefficients from \textsf{IP}-1 must be adapted. For \textsf{IP}-5 until \textsf{IP}-8, the inclusion of all $K$ crops is enforced. Accordingly, their counterparts \textsf{IP}-1 until \textsf{IP}-4 simplify.

For \textsf{IP}-3 the following is proposed:
\begin{subequations}
\label{eq:ILP3}
\begin{align}
\min &\ \ \sum_{d\in\mathcal{D}} \sum_{k\in\mathcal{K}} \sum_{j\in\mathcal{L}} c_{dj}^k x_{dj}^k + \sum_{k\in\mathcal{K}} \sum_{i<j} c_{ij}^k x_{ij}^k -\sum_{l\in\mathcal{L}}\sum_{k\in\mathcal{K}}r_l^k\delta_l^k + \notag \\ & \ \  \sum_{d\in\mathcal{D}} z^d \xi^d + m \gamma \label{eq_ILP3_objFcn}\\
\mathrm{s.t.} &\ \ \sum_{d\in\mathcal{D}} x_{dl}^k + \sum_{i<l} x_{il}^k + \sum_{l<j} x_{lj}^k = 2 \delta_{l}^k,~\forall  l\in\mathcal{L},~\forall k\in\mathcal{K},\label{eq_ILP3_1stcstrt}\\
& \ \ \sum_{k\in\mathcal{K}} \delta_l^k=1,~\forall l\in\mathcal{L}, \label{eq_ILP3_2ndcstrt}\\
& \ \ \sum_{k\in\mathcal{K}}\sum_{j\in\mathcal{L}} x_{dj}^k= 2 p^d,~\forall d\in\mathcal{D}, \\
& \ \ \sum_{d\in\mathcal{D}} \xi^d= 1, \label{eq_ILP3_sumxid_1}\\
& \ \ \xi^d  \leq p^d \leq K \xi^d,~\forall d\in\mathcal{D},\label{eq_ILP3_pd_1}\\
& \ \ p^d  \leq \gamma - (1-\xi^d),~\forall d\in\mathcal{D},\label{eq_ILP3_pd_2}\\
& \ \ p^d  \geq \gamma - K(1-\xi^d),~\forall d\in\mathcal{D},\label{eq_ILP3_pd_3}\\
& \ \ \sum_{i,j\in S^k,i<j} x_{ij}^k \leq |S^k|-1,~S^k\subseteq \mathcal{V}\backslash\{d\},~3\leq |S^k|\leq N-1,\notag\\[-0.3cm]
& \ \ \hspace{4.1cm} ~\forall k\in\mathcal{K},~\forall d\in\mathcal{D}, 
\end{align}
\end{subequations} 
with decision variables 
\begin{subequations}
\begin{align}
& \ \ x_{dj}^k\in\{0,1,2\},~\forall d\in\mathcal{D},~\forall j\in\mathcal{L},~\forall k\in\mathcal{K}, \\
& \ \ x_{ij}^k\in\{0,1\},~0\leq i<j,~\forall k\in\mathcal{K}, \\
& \ \ \delta_l^k\in\{0,1\},~\forall l\in\mathcal{L},~\forall k\in\mathcal{K},  \\
& \ \ 1\leq \gamma \leq K, \\
& \ \ \xi^d\in\{0,1\},~\forall d\in\mathcal{D},\\
& \ \ p^d\in\{0,1,\dots,K\},~\forall d\in\mathcal{D},
\end{align}
\end{subequations} 
Thus, \textsf{IP}-3 has  $N_z = KDL + K \sum_{q=0}^{L-2}L-1-q + KL + 1 + 2D$ integer decision variables. It differs from \textsf{IP}-1 by $D>1$. Therefore, the decision about starting from the optimal depot w.r.t. the given cost criterion is modeled as constraints
\begin{equation}
\sum_{k\in\mathcal{K}} \sum_{j\in\mathcal{L}} x_{dj}^k = 2\gamma \xi^d,~\forall d\in\mathcal{D},\label{eq_ILP3_nonlinConstrt}
\end{equation}
with $x_{dj}^k\in\{0,1,2\}$, $\xi^d\in\{0,1\}$, $1\leq \gamma\leq K$ and $\sum_{d\in\mathcal{D}}\xi^d = 1$. Since \eqref{eq_ILP3_nonlinConstrt} is nonlinear, auxiliary variables $p^d = \gamma \xi^d,~\forall d\in\mathcal{D}$ are introduced. Then, \eqref{eq_ILP3_nonlinConstrt} can be translated to linear inequality constraints \eqref{eq_ILP3_pd_1}, \eqref{eq_ILP3_pd_2} and \eqref{eq_ILP3_pd_3} following \eqref{eq_logic3}. Finally, note that from \eqref{eq_ILP3_sumxid_1} and the definition of $p^d$, the number of decision variables can be reduced by $\gamma$ when substituting $\gamma=\sum_{d\in\mathcal{D}} p^d$ in \eqref{eq:ILP3}.

For \textsf{IP}-4 the following is proposed: 
\begin{subequations}
\label{eq:ILP4}
\begin{align}
\min &\ \  \sum_{d\in\mathcal{D}} \sum_{k\in\mathcal{K}} \sum_{j\in\mathcal{L}} c_{dj}^{k,k^\text{min}} v_{dj}^k + c_{dj}^{k} x_{dj}^k - c_{dj}^{k}v_{dj}^k  + \notag\\
 & \ \ \sum_{k\in\mathcal{K}} \sum_{i<j} c_{ij}^k x_{ij}^k -  \sum_{l\in\mathcal{L}}\sum_{k\in\mathcal{K}}r_l^k\delta_l^k + \sum_{d\in\mathcal{D}} z^d \xi^d \notag\\
 & \ \  \sum_{d\in\mathcal{D}} \sum_{k\in\mathcal{K}} \sum_{j\in\mathcal{L}}  c_{jd}^{k,k^\text{max}} w_{jd}^k + c_{jd}^{k} x_{jd}^k  - c_{jd}^{k} w_{jd}^k + m \gamma \label{eq_ILP4_objFcn}\\
\mathrm{s.t.} &\ \ \sum_{d\in\mathcal{D}} x_{dl}^k + \sum_{i<l} x_{il}^k + \sum_{l<j} x_{lj}^k + \sum_{d\in\mathcal{D}} x_{ld}^k = 2 \delta_{l}^k,~\forall  l\in\mathcal{L}, ~\forall k\in\mathcal{K},\label{eq_ILP4_1stcstrt}\\
& \ \ \sum_{k\in\mathcal{K}} \delta_l^k=1,~\forall l\in\mathcal{L}, \\
& \ \ \sum_{k\in\mathcal{K}}\sum_{j\in\mathcal{L}} x_{dj}^k + x_{jd}^k = 2p^d,~\forall d\in\mathcal{D}, \\
& \ \ \sum_{d\in\mathcal{D}} \xi^d = 1,~\sum_{k\in\mathcal{K}} \tilde{a}^k=1,~\sum_{k\in\mathcal{K}} \tilde{\beta}^k = 1, \\
& \ \ \tilde{\alpha}^0 = \alpha^0,~\tilde{\beta}^{K-1} = \alpha^{K-1},\\
& \ \ \sum_{d\in\mathcal{D}}\sum_{j\in\mathcal{L}}\sum_{k\in\mathcal{K}}v_{dj}^k=1,~\sum_{d\in\mathcal{D}}\sum_{j\in\mathcal{L}}\sum_{k\in\mathcal{K}}w_{jd}^k=1\\
& \ \ \sum_{j\in\mathcal{L}} x_{dj}^k = \xi^d,~\forall d\in\mathcal{D},~\forall k\in\mathcal{K},\\
& \ \ \sum_{j\in\mathcal{L}} x_{jd}^k = \xi^d,~\forall d\in\mathcal{D},~\forall k\in\mathcal{K},\\
& \ \ \xi^d  \leq p^d \leq K,~\forall d\in\mathcal{D},\label{eq_ILP4_pd_1}\\
& \ \ p^d  \leq \gamma - (1-\xi^d),~\forall d\in\mathcal{D},\label{eq_ILP4_pd_2}\\
& \ \ p^d  \geq \gamma - K(1-\xi^d),~\forall d\in\mathcal{D},\label{eq_ILP4_pd_3}\\
& \ \ 1 - \sum_{l\in\mathcal{L}} \delta_l^k \leq 1-\alpha^k,~\forall k\in\mathcal{K},\label{eq_ILP4_alpha1}\\
& \ \ 1 - \sum_{l\in\mathcal{L}} \delta_l^k \geq \epsilon + (-L+1 - \epsilon)\alpha^k,~\forall k\in\mathcal{K}, \label{eq_ILP4_alpha2}\\
& \ \ \alpha^k + (1-\sum_{\tau=0}^{k-1}\tilde{\alpha}^{\tau}) - \tilde{\alpha}^k  \leq 1,~\forall k=1,\dots,K-1,\label{eq_ILP4_tildealphak_1}\\
& \ \ \tilde{\alpha}^k \leq \alpha^k,~\tilde{\alpha}^k \leq 1-\sum_{\tau=0}^{k-1}\tilde{\alpha}^{\tau},~\forall k=1,\dots,K-1. \label{eq_ILP4_tildealphak_2}\\
& \ \ \alpha^{K-2-k} + (1-\sum_{\tau=1}^{1+k}\tilde{\beta}^{K-2-k+\tau}) - \tilde{\beta}^{K-2-k}  \leq 1,\notag\\[-0.4cm]
& \ \ \hspace{3.8cm}~\forall k=0,\dots,K-2,\label{eq_ILP4_tildebetak_1}\\
& \ \ \tilde{\beta}^{K-2-k} \leq \alpha^{K-2-k},~\forall k=0,\dots,K-2,\label{eq_ILP4_tildebetak_2}\\
& \ \ \tilde{\beta}^{K-2-k} \leq 1-\sum_{\tau=1}^{1+k}\tilde{\beta}^{K-2-k+\tau},~\forall k=0,\dots,K-2,\label{eq_ILP4_tildebetak_3}\\ 
& \ \  \tilde{\alpha}^k + x_{dj}^k - v_{dj}^k \leq 1,~v_{dj}^k\leq \tilde{\alpha}^k,~v_{dj}^{k}\leq x_{dj}^k,\notag\\[-0.cm]
& \ \ \hspace{2.9cm} ~\forall d\in\mathcal{D},~\forall j\in\mathcal{L},~\forall k\in\mathcal{K},\\
& \ \  \tilde{\beta}^k + x_{jd}^k - w_{jd}^k \leq 1,~w_{jd}^k\leq \tilde{\beta}^k,~w_{jd}^{k}\leq x_{jd}^k,\notag\\[-0.cm]
& \ \ \hspace{2.8cm} ~\forall d\in\mathcal{D},~\forall j\in\mathcal{L},~\forall k\in\mathcal{K},\\
& \ \ \sum_{i,j\in S^k,~i<j} x_{ij}^k \leq |S^k|-1,~S^k\subseteq \mathcal{V}\backslash\{d\},~3\leq |S^k|\leq N-1,\notag\\[-0.3cm]
& \ \ \hspace{3.5cm} ~\forall k\in\mathcal{K},~\forall d \in\mathcal{D},
\end{align}
\end{subequations} 
with decision variables
\begin{subequations}
\label{eq:ILP4_variables}
\begin{align}
& \ \ x_{dj}^k\in\{0,1\},~\forall d\in\mathcal{D},~\forall j\in\mathcal{L},~\forall k\in\mathcal{K}, \\
& \ \ x_{ij}^k\in\{0,1\},~0\leq i<j,~\forall k\in\mathcal{K}, \\
& \ \ \delta_l^k\in\{0,1\},~\forall l\in\mathcal{L},~\forall k\in\mathcal{K},  \\
& \ \ 1\leq \gamma \leq K, \\
& \ \ \xi^d\in\{0,1\},~\forall d\in\mathcal{D},\\
& \ \ p^d\in\{0,1,\dots,K\},~\forall d\in\mathcal{D},\\
& \ \ \alpha^k,\tilde{\alpha}^k,\tilde{\beta}^k\in\{0,1\},~\forall k\in\mathcal{K},\\
& \ \ v_{dj}^k,w_{jd}^k\in\{0,1\},~\forall d\in\mathcal{D},~\forall j\in\mathcal{L},~\forall k\in\mathcal{K},\\
& \ \ x_{jd}^k\in\{0,1\},~\forall d\in\mathcal{D},~\forall j\in\mathcal{L},~\forall k\in\mathcal{K}.
\end{align}
\end{subequations} 
 For the formulation of crop- and depot-dependent cost coefficients, the minimum and maximum active crop-indices need to be identified. Let therefore $\alpha^k\in\{0,1\}$ indicate if crop $k$ is active in the sense of $\alpha^k=1$ if $\sum_{l\in\mathcal{L}} \delta_l^k\geq 1$. By \eqref{eq_logic1}, this translates to \eqref{eq_ILP4_alpha1} and \eqref{eq_ILP4_alpha2}. 
We introduced auxiliary variables $\tilde{\alpha}^k,\tilde{\beta}^k \in\{0,1\}$ indicating if crop $k$ is the smallest- or largest-indexed active crop, respectively ($k=k^\text{min}$ and $k=k^\text{max}$). It holds that $\sum_{k\in\mathcal{K}} \tilde{\alpha}^k = 1$ and $\sum_{k\in\mathcal{K}} \tilde{\beta}^k=1$. We then derive the nonlinear relations $\tilde{\alpha}^0= \alpha^0$, $\tilde{\alpha}^1 = \alpha^1 (1-\tilde{\alpha}^0)$, $\tilde{\alpha}^2 = \alpha^2(1-\tilde{\alpha}^1-\tilde{\alpha}^0)$, \dots, which can be translated to
\begin{subequations}
\begin{align}
&\tilde{\alpha}^0 = \alpha^0,\\
&\alpha^k + (1-\sum_{\tau=0}^{k-1}\tilde{\alpha}^{\tau}) - \tilde{\alpha}^k  \leq 1,~\forall k=1,\dots,K-1,\\
&\tilde{\alpha}^k \leq \alpha^k,~\tilde{\alpha}^k \leq 1-\sum_{\tau=0}^{k-1}\tilde{\alpha}^{\tau},~\forall k=1,\dots,K-1.
\end{align}
\end{subequations}
Similarly, starting the iteration from highest $k=K-1$ with $\tilde{\beta}^{K-1} = \alpha^{K-1}$, we can derive nonlinear relations for $\tilde{\beta}^k$ to ultimately obtain \eqref{eq_ILP4_tildebetak_1}, \eqref{eq_ILP4_tildebetak_2} and \eqref{eq_ILP4_tildebetak_3}. Suppose the \emph{path-dependent} part of the cost function taking the nonlinear form $\sum_{d\in\mathcal{D}} \sum_{k\in\mathcal{K}} \sum_{j\in\mathcal{L}} \left( c_{dj}^{k,k^\text{min}} \tilde{\alpha}^k + c_{dj}^{k}(1-\tilde{\alpha}^k) \right) x_{dj}^k+$\newline  $\sum_{d\in\mathcal{D}} \sum_{k\in\mathcal{K}} \sum_{j\in\mathcal{L}} \left( c_{jd}^{k,k^\text{max}} \tilde{\beta}^k + c_{jd}^{k}(1-\tilde{\beta}^k) \right) x_{jd}^k $, with $c_{dj}^{k,k^\text{min}}\geq 0$ and  further $c_{jd}^{k,k^\text{max}}\geq 0$ denoting cost-coefficients that are distinct for the first (i.e., $k=k^\text{min}$ or $\tilde{\alpha}^k=1$) and last (i.e., $k=k^\text{max}$ or $\tilde{\beta}^k=1$) crop-route, respectively. Then, auxiliary variables $v_{dj}^k\in\{0,1\}$ and $w_{jd}^k\in\{0,1\}$ need to be introduced with $\sum_{d\in\mathcal{D}}\sum_{j\in\mathcal{L}}\sum_{k\in\mathcal{K}}v_{dj}^k=1$ and $\sum_{d\in\mathcal{D}}\sum_{j\in\mathcal{L}}\sum_{k\in\mathcal{K}}w_{jd}^k=1$. They are related according to $v_{dj}^k = \tilde{\alpha}^k x_{dj}^k$ and $w_{jd}^k=\tilde{\beta}^k x_{jd}^k,~\forall d\in\mathcal{D},~j\in\mathcal{L},~k\in\mathcal{K}$, and can be translated to integer linear inequalities according to \eqref{eq_logic2}. The objective function part above can now be expressed linearly dependent on decision variables, see \eqref{eq_ILP4_objFcn}.

\subsection{Main Algorithm\label{subsec_MainAlg}}

The main algorithm of this paper is summarised as follows:

\vspace{0.4cm}
\hspace{-0.3cm}\begin{tabular}{ l l }
\hline \\[-8pt]
 \multicolumn{2}{l}{\textbf{Algorithm 1}: \textsf{CApR}-$n$}\\[3pt]
\hline
1: & \textbf{Input}: $c_{dj}^k$, $c_{ij}^k$, $c_{dj}^{k,k_\text{min}}$, $c_{jd}^{k,k_\text{max}}$, $c_{jd}^k$, $r_l^k$, $m$, $\{z^d\}_{d=0}^{D-1}$ and $\tilde{k}$.\\
2: & \textbf{Clustering}:\\
& - cluster $L$ fields according to an arbitrary criterion, \\ & e.g., spatially based according to $\tilde{k}$-means.\\
& - let the sets of fields associated with each cluster be\\
& denoted by $\mathcal{L}_{z,\xi}\subset \mathcal{L},~\forall \xi=0,\dots,\tilde{k}-1$.\\
& - let the set of clusters be denoted by $\mathcal{L}_{z}$ with $|\mathcal{L}_{z}|=\tilde{k}$.\\
& - assign a coordinate to each cluster (centroids).\\
& - compute $c_{dj_z}^k$, $c_{i_z j_z}^k$, $c_{dj_z}^{k,k_\text{min}}$, $c_{j_z d}^{k,k_\text{max}}$, $c_{j_z d}^k,~\forall i_z,j_z\in\mathcal{L}_{z}$.\\
& - compute $r_{l_z}^k = \sum_{l\in\mathcal{L}_{z,\xi}} r_l^k,~\forall l_z\in\mathcal{L}_{z},~\xi=0,\dots,\tilde{k}-1$.\\[5pt]
3: & \textbf{Integer Programming}  (\textsf{IP}-$n$):\\
& - solve \textsf{IP}-$n$ for the clustering result of Step 2, re-\\
&  placing $\mathcal{L}$ by $\mathcal{L}_z$ and cost coefficients accordingly.\\
& - let the resulting set of active crops and optimal \\
&  basis depot be denoted $\mathcal{M}^\star\subseteq \mathcal{K}$ and $d^\star\in\mathcal{D}$.\\
& - let $\mathcal{C}^k$ denote the sequence of clusters $\forall k\in\mathcal{M}^\star$,  \\
& where every sequence starts and ends at $d^\star\in\mathcal{D}$.\\[5pt]
4: & \textbf{From Cluster- to Field-sequences}:\\
& FOR $k\in\mathcal{M}^\star$:\\
& \hspace*{0.4cm} - define $\mathcal{C}^{k,1}=\mathcal{C}^k$ and $\mathcal{C}^{k,2}=\text{flip}\left( \mathcal{C}^k \right)$.\\
& \hspace*{0.4cm} FOR $i=1,2$:\\
& \hspace*{0.8cm} FOR $\mathcal{C}^{k,i}$:\\
& \hspace*{1.2cm} - find the closest fields between any pair \\ 
& \hspace*{1.2cm} of consecutive clusters $c^{(t)},c^{(t+1)}\in\mathcal{C}^{k,i}$  \\
& \hspace*{1.2cm}  within the $\mathcal{C}^{k,i}$-tour, where $t=0,\dots,|\mathcal{C}^{k,i}|$.\\
& \hspace*{1.2cm} - let the two fields associated with each \\
& \hspace*{1.2cm} cluster $c^{(t)}$ be  denoted by $s^{(t)}$ and $e^{(t)}$.\\ 
& \hspace*{1.2cm} - for each cluster $c^{(t)},\forall t$, solve a TSP \\
& \hspace*{1.2cm} connecting $s^{(t)}$ and $e^{(t)}$ to obtain a \\
& \hspace*{1.2cm} field-sequence $f^{(t)} = \{s^{(t)},\dots,e^{(t)}\}$.\\
& \hspace*{1.2cm} - concatenate field-sequences to crop-tour \\
& \hspace*{1.2cm} $\mathcal{F}^{k,i} = \{f^{(0)},\dots, f^{(|\mathcal{C}^{k,i}|)}\}$  and compute its \\
& \hspace*{1.2cm} path length $d^{k,i}$.\\
& \hspace*{0.4cm} IF $d^{k,1}<d^{k,2}$: $\mathcal{F}^{k,\star}=\mathcal{F}^{k,1}$, else $\mathcal{F}^{k,\star}=\mathcal{F}^{k,2}$.\\[5pt]
5: & \textbf{Output}:\\
& - set of active crops $\mathcal{M}^\star$, and basis depot $d^\star\in\mathcal{D}$.\\
& - crop assignment to fields, $\delta_l^{k,\star},~\forall l\in\mathcal{L},\forall k\in\mathcal{M}^\star$.\\
& - crop-tour $\mathcal{F}^{k,\star},~\forall k\in\mathcal{M}^\star$.\\
& - monetary result $J^{\textsf{CApR}-n}$.\\[3pt]\hline\\[3pt]
%
\end{tabular}

The name is derived from its purpose: \emph{Crop Assignment plus Routing} (\textsf{CApR}). See Figure \ref{fig:alg_CApR} for visualisation. The reversing of a list or sequence of elements is denoted by the $\text{flip}(\cdot)$-operator. The monetary profit is denoted by $J^{\textsf{CApR}-n}$ and is always composed of all of (a) revenues from growing crops, (b) accumulated edge costs, (c) relevant depot maintenance costs and (d) fixed crop costs. Thus, a large $J^{\textsf{CApR}-n}$ is beneficial. For clarification, for $\tilde{k}=L$ the profit $J^{\textsf{CApR}-n}$ is equal to the negative objective value of \textsf{IP}-$n$, $J^{\textsf{IP}-n}$, when all of (a) to (d) are accounted for. Analysis is provided next.

\newlength\figureheight
\newlength\figurewidth 
\setlength\figureheight{5.55cm}
\setlength\figurewidth{5.55cm}
\begin{figure*}
\centering%
\begin{tikzpicture}

\definecolor{color0}{rgb}{1,1,0}

\begin{axis}[
xlabel={\small{$X$ [km]}},xlabel near ticks,
ylabel={\small{$Y$ [km]}}, ylabel near ticks,
xmin=-22, xmax=24.5,
ymin=-26, ymax=29.6,
axis on top,
width=\figurewidth,
height=\figureheight
]
\addplot [very thick, black, mark=*, mark size=3, mark options={solid,draw=black}, only marks]
table {%
5.54912717210443 -9.17976854024501
-5.36075647529545 -4.69183630537719
-7.13690413623632 4.97980950089336
};
\addplot [very thick, green!50.196078431372548!black, mark=*, mark size=1, mark options={solid}, only marks]
table {%
-1.81141816833373 -14.7536174356505
-8.03826228605862 -4.41702565665195
7.41410809258517 4.84885103731669
10.7776825823822 8.43493264112755
-17.2176624878364 9.71560159200989
-5.87182855358123 10.2605100264461
-11.1477491373684 0.636192570530905
-20.6200611942363 -3.4495698771771
-2.45562512969019 8.20397380459398
-3.30015579547515 -15.4415652878383
-10.1283971789637 9.99821818315607
-6.14207425605862 -1.79381909566184
-0.048513916390766 3.45066694356219
-3.45823196068543 -3.05501662981954
-7.63556351876465 20.137805506487
1.95147271965853 10.8328700425764
5.54847141398477 -0.0346969158307768
-9.60987375516025 -0.48188111659666
-1.95580284414657 -5.67635720788008
13.6567316035884 -16.610497068477
-1.19596435637135 -5.34569035326587
-9.11756264643322 -12.478354569907
0.83496501788204 -8.3644742808282
14.195518265761 1.13563781528158
7.43204266618025 -11.3768849099122
-5.96166247088939 -3.62475290243678
1.89203395151815 11.5478548640834
-3.74741403573997 -2.94613340487459
21.7340097735993 -4.0799871411908
-11.5251844141959 10.4266580162947
-10.0009381868691 25.7990314276531
11.8416707514425 -0.774527587990311
-11.5462196445152 -1.73188666151523
13.5144882290728 -9.45889421572373
7.80550327072996 13.4992160670252
13.3751420065552 2.88357960720043
13.9721262004041 -14.7087307855056
-4.14575446613101 -1.93192237625023
10.5426318584968 -1.66904393614652
-9.24595060001878 3.39247802508048
-13.8596192183418 9.24248650224519
8.31054036231326 -18.3461646132472
13.0682511786582 -23.6936636972141
1.38643881426368 23.030318975339
2.49050693879937 -10.2823001982673
1.10361795370018 8.13336720243664
1.55709347332873 -21.531526174017
-14.2732766350511 -12.9716742575041
-1.6713894155691 12.0007273925418
3.31056229212735 2.8134229350964
};
\draw[] (axis cs:5.54912717210443,-9.17976854024501) -- (axis cs:5.54912717210443,-9.17976854024501);
\node at (axis cs:5.54912717210443,-9.17976854024501)[
  scale=0.8,
  anchor=base west,
  text=black,
  rotate=0.0
]{ D0};
\draw[] (axis cs:-5.36075647529545,-4.69183630537719) -- (axis cs:-5.36075647529545,-4.69183630537719);
\node at (axis cs:-5.36075647529545,-4.69183630537719)[
  scale=0.8,
  anchor=base west,
  text=black,
  rotate=0.0
]{ D1};
\draw[] (axis cs:-7.13690413623632,4.97980950089336) -- (axis cs:-7.13690413623632,4.97980950089336);
\node at (axis cs:-7.13690413623632,4.97980950089336)[
  scale=0.8,
  anchor=base west,
  text=black,
  rotate=0.0
]{ D2};
\node at (axis cs:18.5,22.25)[
  scale=1.5,
  anchor=base west,
  text=black,
  rotate=0.0
]{ \textsf{1}};
\end{axis}

\end{tikzpicture}\hspace{-0.1cm}
\begin{tikzpicture}

\definecolor{color1}{rgb}{0.46078431372549,0.0615609061339428,0.999525719713366}
\definecolor{color0}{rgb}{1,1,0}
\definecolor{color3}{rgb}{0.0372549019607843,0.664540178707858,0.934679767321111}
\definecolor{color2}{rgb}{0.249019607843137,0.384105749171926,0.980634770468978}
\definecolor{color5}{rgb}{0.394117647058823,0.986200747353403,0.763398282741103}
\definecolor{color4}{rgb}{0.182352941176471,0.878081248083698,0.859799851448372}
\definecolor{color7}{rgb}{0.817647058823529,0.878081248083698,0.510631193180907}
\definecolor{color6}{rgb}{0.605882352941177,0.986200747353403,0.645928062486787}
\definecolor{color9}{rgb}{1,0.384105749171926,0.195845467007167}
\definecolor{color8}{rgb}{1,0.664540178707858,0.355490833300318}
\definecolor{color10}{rgb}{1,0.0615609061339432,0.0307950585561705}

\begin{axis}[
xlabel={\small{$X$ [km]}},xlabel near ticks,
yticklabels={\empty},
xmin=-22, xmax=24.5,
ymin=-26, ymax=29.6,
axis on top,
width=\figurewidth,
height=\figureheight
]
\addplot [very thick, black, mark=*, mark size=3, mark options={solid,draw=black}, only marks]
table {%
5.54912717210443 -9.17976854024501
-5.36075647529545 -4.69183630537719
-7.13690413623632 4.97980950089336
};
\addplot [very thick, color1, mark=+, mark size=3, mark options={solid}, only marks]
table {%
-5.4166876304567 22.9890519698264
};
\addplot [very thick, color1, mark=*, mark size=1, mark options={solid}, only marks]
table {%
-7.63556351876465 20.137805506487
-10.0009381868691 25.7990314276531
1.38643881426368 23.030318975339
};
\addplot [very thick, color2, mark=+, mark size=3, mark options={solid}, only marks]
table {%
-12.4339708662598 -0.326933411935521
};
\addplot [very thick, color2, mark=*, mark size=1, mark options={solid}, only marks]
table {%
-11.1477491373684 0.636192570530905
-20.6200611942363 -3.4495698771771
-9.60987375516025 -0.48188111659666
-11.5462196445152 -1.73188666151523
-9.24595060001878 3.39247802508048
};
\addplot [very thick, color2, mark=*, mark size=1, mark options={solid}, only marks]
table {%
-11.1477491373684 0.636192570530905
-20.6200611942363 -3.4495698771771
-9.60987375516025 -0.48188111659666
-11.5462196445152 -1.73188666151523
-9.24595060001878 3.39247802508048
};
\addplot [very thick, color3, mark=+, mark size=3, mark options={solid}, only marks]
table {%
2.49440513158438 11.9701670915567
};
\addplot [very thick, color3, mark=*, mark size=1, mark options={solid}, only marks]
table {%
1.95147271965853 10.8328700425764
1.89203395151815 11.5478548640834
7.80550327072996 13.4992160670252
-1.6713894155691 12.0007273925418
};
\addplot [very thick, color4, mark=+, mark size=3, mark options={solid}, only marks]
table {%
-0.792357470766928 6.57238818242706
};
\addplot [very thick, color4, mark=*, mark size=1, mark options={solid}, only marks]
table {%
-5.87182855358123 10.2605100264461
-2.45562512969019 8.20397380459398
-0.048513916390766 3.45066694356219
1.10361795370018 8.13336720243664
3.31056229212735 2.8134229350964
};
\addplot [very thick, color5, mark=+, mark size=3, mark options={solid}, only marks]
table {%
11.6590300400362 -15.6991392150133
};
\addplot [very thick, color5, mark=*, mark size=1, mark options={solid}, only marks]
table {%
13.6567316035884 -16.610497068477
7.43204266618025 -11.3768849099122
13.5144882290728 -9.45889421572373
13.9721262004041 -14.7087307855056
8.31054036231326 -18.3461646132472
13.0682511786582 -23.6936636972141
};
\addplot [very thick, color5, mark=*, mark size=1, mark options={solid}, only marks]
table {%
13.6567316035884 -16.610497068477
7.43204266618025 -11.3768849099122
13.5144882290728 -9.45889421572373
13.9721262004041 -14.7087307855056
8.31054036231326 -18.3461646132472
13.0682511786582 -23.6936636972141
};
\addplot [very thick, color6, mark=+, mark size=3, mark options={solid}, only marks]
table {%
9.53160719080081 2.89272448673347
};
\addplot [very thick, color6, mark=*, mark size=1, mark options={solid}, only marks]
table {%
7.41410809258517 4.84885103731669
10.7776825823822 8.43493264112755
5.54847141398477 -0.0346969158307768
13.3751420065552 2.88357960720043
10.5426318584968 -1.66904393614652
};
\addplot [very thick, color7, mark=+, mark size=3, mark options={solid}, only marks]
table {%
-13.1827158248344 9.84574107342646
};
\addplot [very thick, color7, mark=*, mark size=1, mark options={solid}, only marks]
table {%
-17.2176624878364 9.71560159200989
-10.1283971789637 9.99821818315607
-11.5251844141959 10.4266580162947
-13.8596192183418 9.24248650224519
};
\addplot [very thick, color8, mark=+, mark size=3, mark options={solid}, only marks]
table {%
-3.75668907313321 -4.12835465640768
};
\addplot [very thick, color8, mark=*, mark size=1, mark options={solid}, only marks]
table {%
-8.03826228605862 -4.41702565665195
-6.14207425605862 -1.79381909566184
-3.45823196068543 -3.05501662981954
-1.95580284414657 -5.67635720788008
-1.19596435637135 -5.34569035326587
0.83496501788204 -8.3644742808282
-5.96166247088939 -3.62475290243678
-3.74741403573997 -2.94613340487459
-4.14575446613101 -1.93192237625023
};

\addplot [very thick, color9, mark=+, mark size=3, mark options={solid}, only marks]
table {%
-4.07580213886085 -14.5765063205307
};
\addplot [very thick, color9, mark=*, mark size=1, mark options={solid}, only marks]
table {%
-1.81141816833373 -14.7536174356505
-3.30015579547515 -15.4415652878383
-9.11756264643322 -12.478354569907
2.49050693879937 -10.2823001982673
1.55709347332873 -21.531526174017
-14.2732766350511 -12.9716742575041
};

\addplot [very thick, color10, mark=+, mark size=3, mark options={solid}, only marks]
table {%
15.9237329302676 -1.23962563796651
};
\addplot [very thick, color10, mark=*, mark size=1, mark options={solid}, only marks]
table {%
14.195518265761 1.13563781528158
21.7340097735993 -4.0799871411908
11.8416707514425 -0.774527587990311
};
\draw[] (axis cs:-5.4166876304567,22.9890519698264) -- (axis cs:-5.4166876304567,22.9890519698264);
\node at (axis cs:-5.4166876304567,22.9890519698264)[
  scale=0.7,
  anchor=base west,
  text=color1,
  rotate=0.0
]{ C0};
\draw[] (axis cs:-12.4339708662598,-0.326933411935521) -- (axis cs:-12.4339708662598,-0.326933411935521);
\node at (axis cs:-12.4339708662598,-0.326933411935521)[
  scale=0.7,
  anchor=base west,
  text=color2,
  rotate=0.0
]{ C1};
\draw[] (axis cs:2.49440513158438,11.9701670915567) -- (axis cs:2.49440513158438,11.9701670915567);
\node at (axis cs:2.49440513158438,11.9701670915567)[
  scale=0.7,
  anchor=base west,
  text=color3,
  rotate=0.0
]{ C2};
\draw[] (axis cs:-0.792357470766928,6.57238818242706) -- (axis cs:-0.792357470766928,6.57238818242706);
\node at (axis cs:-0.792357470766928,6.57238818242706)[
  scale=0.7,
  anchor=base west,
  text=color4,
  rotate=0.0
]{ C3};
\draw[] (axis cs:11.6590300400362,-15.6991392150133) -- (axis cs:11.6590300400362,-15.6991392150133);
\node at (axis cs:11.6590300400362,-15.6991392150133)[
  scale=0.7,
  anchor=base west,
  text=color5,
  rotate=0.0
]{ C4};
\draw[] (axis cs:9.53160719080081,2.89272448673347) -- (axis cs:9.53160719080081,2.89272448673347);
\node at (axis cs:9.53160719080081,2.89272448673347)[
  scale=0.7,
  anchor=base west,
  text=color6,
  rotate=0.0
]{ C5};
\draw[] (axis cs:-13.1827158248344,9.84574107342646) -- (axis cs:-13.1827158248344,9.84574107342646);
\node at (axis cs:-13.1827158248344,9.84574107342646)[
  scale=0.7,
  anchor=base west,
  text=color7,
  rotate=0.0
]{ C6};
\draw[] (axis cs:-3.75668907313321,-4.12835465640768) -- (axis cs:-3.75668907313321,-4.12835465640768);
\node at (axis cs:-3.75668907313321,-4.12835465640768)[
  scale=0.7,
  anchor=base west,
  text=color8,
  rotate=0.0
]{ C7};
\draw[] (axis cs:-4.07580213886085,-14.5765063205307) -- (axis cs:-4.07580213886085,-14.5765063205307);
\node at (axis cs:-4.07580213886085,-14.5765063205307)[
  scale=0.7,
  anchor=base west,
  text=color9,
  rotate=0.0
]{ C8};
\draw[] (axis cs:15.9237329302676,-1.23962563796651) -- (axis cs:15.9237329302676,-1.23962563796651);
\node at (axis cs:15.9237329302676,-1.23962563796651)[
  scale=0.7,
  anchor=base west,
  text=color10,
  rotate=0.0
]{ C9};
\draw[] (axis cs:5.54912717210443,-9.17976854024501) -- (axis cs:5.54912717210443,-9.17976854024501);
\node at (axis cs:5.54912717210443,-9.17976854024501)[
  scale=0.8,
  anchor=base west,
  text=black,
  rotate=0.0
]{ D0};
\draw[] (axis cs:-5.36075647529545,-4.69183630537719) -- (axis cs:-5.36075647529545,-4.69183630537719);
\node at (axis cs:-5.36075647529545,-4.69183630537719)[
  scale=0.8,
  anchor=base west,
  text=black,
  rotate=0.0
]{ D1};
\draw[] (axis cs:-7.13690413623632,4.97980950089336) -- (axis cs:-7.13690413623632,4.97980950089336);
\node at (axis cs:-7.13690413623632,4.97980950089336)[
  scale=0.8,
  anchor=base west,
  text=black,
  rotate=0.0
]{ D2};
\node at (axis cs:18.5,22.25)[
  scale=1.5,
  anchor=base west,
  text=black,
  rotate=0.0
]{ \textsf{2}};
\end{axis}

\end{tikzpicture}\hspace{-0.1cm}\input{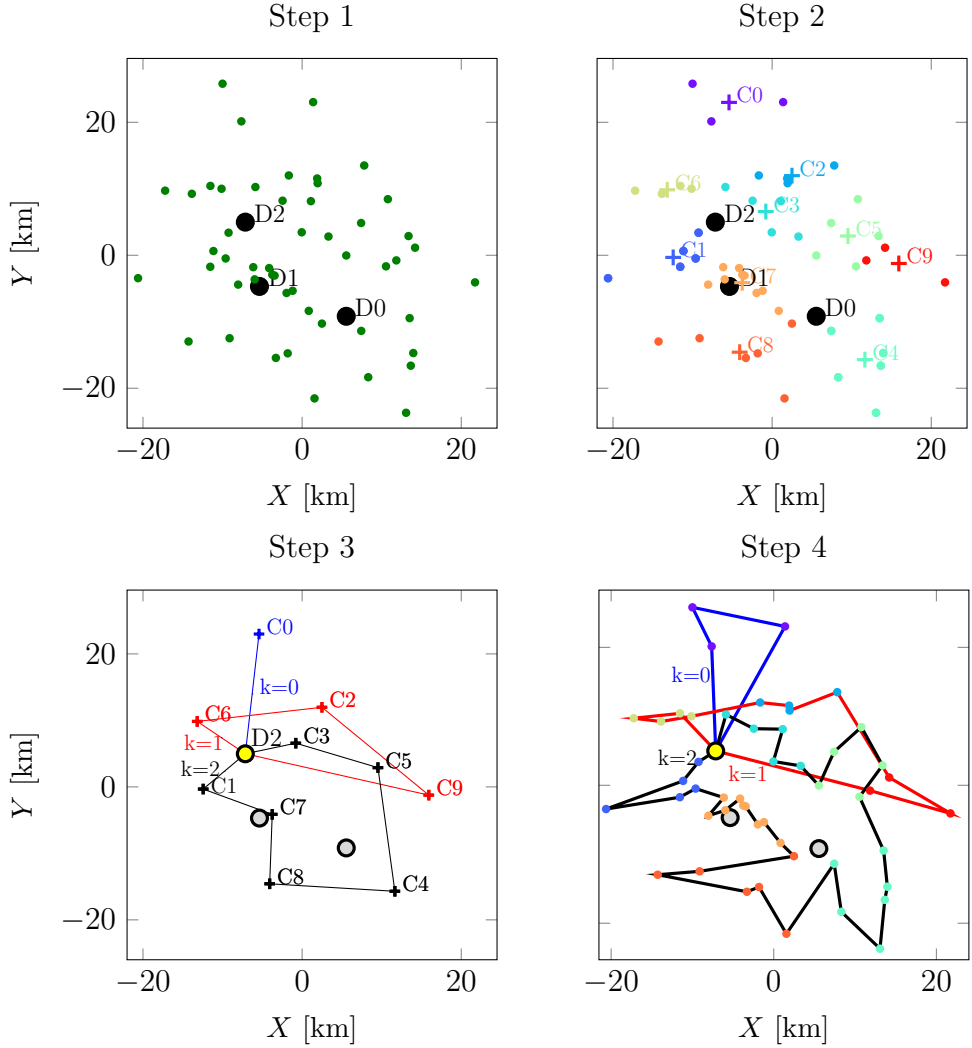}\hspace{-0.1cm}
\begin{tikzpicture}

\definecolor{color1}{rgb}{0.46078431372549,0.0615609061339428,0.999525719713366}
\definecolor{color0}{rgb}{1,1,0}
\definecolor{color3}{rgb}{0.0372549019607843,0.664540178707858,0.934679767321111}
\definecolor{color2}{rgb}{0.249019607843137,0.384105749171926,0.980634770468978}
\definecolor{color5}{rgb}{0.394117647058823,0.986200747353403,0.763398282741103}
\definecolor{color4}{rgb}{0.182352941176471,0.878081248083698,0.859799851448372}
\definecolor{color7}{rgb}{0.817647058823529,0.878081248083698,0.510631193180907}
\definecolor{color6}{rgb}{0.605882352941177,0.986200747353403,0.645928062486787}
\definecolor{color9}{rgb}{1,0.384105749171926,0.195845467007167}
\definecolor{color8}{rgb}{1,0.664540178707858,0.355490833300318}
\definecolor{color10}{rgb}{1,0.0615609061339432,0.0307950585561705}

\begin{axis}[
xlabel={\small{$X$ [km]}}, xlabel near ticks,
yticklabels={\empty},
xmin=-21.5, xmax=24.0,
ymin=-25.3, ymax=28.3,
axis on top,
width=\figurewidth,
height=\figureheight
]
\addplot [very thick, color0, mark=*, mark size=3, mark options={solid,draw=black}, only marks]
table {%
-7.13690413623632 4.97980950089336
};
\addplot [very thick, black!15, mark=*, mark size=3, mark options={solid,draw=black}, only marks]
table {%
5.54912717210443 -9.17976854024501
-5.36075647529545 -4.69183630537719
};
\addplot [very thick, color1, mark=*, mark size=1, mark options={solid}, only marks]
table {%
-7.63556351876465 20.137805506487
-10.0009381868691 25.7990314276531
1.38643881426368 23.030318975339
};
\addplot [very thick, color2, mark=*, mark size=1, mark options={solid}, only marks]
table {%
-11.1477491373684 0.636192570530905
-20.6200611942363 -3.4495698771771
-9.60987375516025 -0.48188111659666
-11.5462196445152 -1.73188666151523
-9.24595060001878 3.39247802508048
};
\addplot [very thick, color3, mark=*, mark size=1, mark options={solid}, only marks]
table {%
1.95147271965853 10.8328700425764
1.89203395151815 11.5478548640834
7.80550327072996 13.4992160670252
-1.6713894155691 12.0007273925418
};
\addplot [very thick, color4, mark=*, mark size=1, mark options={solid}, only marks]
table {%
-5.87182855358123 10.2605100264461
-2.45562512969019 8.20397380459398
-0.048513916390766 3.45066694356219
1.10361795370018 8.13336720243664
3.31056229212735 2.8134229350964
};
\addplot [very thick, color5, mark=*, mark size=1, mark options={solid}, only marks]
table {%
13.6567316035884 -16.610497068477
7.43204266618025 -11.3768849099122
13.5144882290728 -9.45889421572373
13.9721262004041 -14.7087307855056
8.31054036231326 -18.3461646132472
13.0682511786582 -23.6936636972141
};
\addplot [very thick, color6, mark=*, mark size=1, mark options={solid}, only marks]
table {%
7.41410809258517 4.84885103731669
10.7776825823822 8.43493264112755
5.54847141398477 -0.0346969158307768
13.3751420065552 2.88357960720043
10.5426318584968 -1.66904393614652
};
\addplot [very thick, color7, mark=*, mark size=1, mark options={solid}, only marks]
table {%
-17.2176624878364 9.71560159200989
-10.1283971789637 9.99821818315607
-11.5251844141959 10.4266580162947
-13.8596192183418 9.24248650224519
};
\addplot [very thick, color8, mark=*, mark size=1, mark options={solid}, only marks]
table {%
-8.03826228605862 -4.41702565665195
-6.14207425605862 -1.79381909566184
-3.45823196068543 -3.05501662981954
-1.95580284414657 -5.67635720788008
-1.19596435637135 -5.34569035326587
0.83496501788204 -8.3644742808282
-5.96166247088939 -3.62475290243678
-3.74741403573997 -2.94613340487459
-4.14575446613101 -1.93192237625023
};
\addplot [very thick, color9, mark=*, mark size=1, mark options={solid}, only marks]
table {%
-1.81141816833373 -14.7536174356505
-3.30015579547515 -15.4415652878383
-9.11756264643322 -12.478354569907
2.49050693879937 -10.2823001982673
1.55709347332873 -21.531526174017
-14.2732766350511 -12.9716742575041
};
\addplot [very thick, color10, mark=*, mark size=1, mark options={solid}, only marks]
table {%
14.195518265761 1.13563781528158
21.7340097735993 -4.0799871411908
11.8416707514425 -0.774527587990311
};
\addplot [very thick, blue]
table {%
-7.13690413623632 4.97980950089336
-7.63556351876465 20.137805506487
-10.0009381868691 25.7990314276531
1.38643881426368 23.030318975339
-7.13690413623632 4.97980950089336
};
\addplot [very thick, red]
table {%
-7.13690413623632 4.97980950089336
11.8416707514425 -0.774527587990311
21.7340097735993 -4.0799871411908
14.195518265761 1.13563781528158
7.80550327072996 13.4992160670252
1.95147271965853 10.8328700425764
1.89203395151815 11.5478548640834
-1.6713894155691 12.0007273925418
-10.1283971789637 9.99821818315607
-13.8596192183418 9.24248650224519
-17.2176624878364 9.71560159200989
-11.5251844141959 10.4266580162947
-7.13690413623632 4.97980950089336
};
\addplot [very thick, black]
table {%
-7.13690413623632 4.97980950089336
-9.24595060001878 3.39247802508048
-11.1477491373684 0.636192570530905
-20.6200611942363 -3.4495698771771
-11.5462196445152 -1.73188666151523
-9.60987375516025 -0.48188111659666
-6.14207425605862 -1.79381909566184
-8.03826228605862 -4.41702565665195
-5.96166247088939 -3.62475290243678
-4.14575446613101 -1.93192237625023
-3.74741403573997 -2.94613340487459
-3.45823196068543 -3.05501662981954
-1.95580284414657 -5.67635720788008
-1.19596435637135 -5.34569035326587
0.83496501788204 -8.3644742808282
2.49050693879937 -10.2823001982673
-9.11756264643322 -12.478354569907
-14.2732766350511 -12.9716742575041
-3.30015579547515 -15.4415652878383
-1.81141816833373 -14.7536174356505
1.55709347332873 -21.531526174017
7.43204266618025 -11.3768849099122
8.31054036231326 -18.3461646132472
13.0682511786582 -23.6936636972141
13.6567316035884 -16.610497068477
13.9721262004041 -14.7087307855056
13.5144882290728 -9.45889421572373
10.5426318584968 -1.66904393614652
13.3751420065552 2.88357960720043
10.7776825823822 8.43493264112755
7.41410809258517 4.84885103731669
5.54847141398477 -0.0346969158307768
3.31056229212735 2.8134229350964
-0.048513916390766 3.45066694356219
1.10361795370018 8.13336720243664
-2.45562512969019 8.20397380459398
-5.87182855358123 10.2605100264461
-7.13690413623632 4.97980950089336
};
\draw[] (axis cs:-7.38623382750048,12.5588075036902) -- (axis cs:-7.38623382750048,12.5588075036902);
\node at (axis cs:-13.5,15)[
  scale=0.7,
  anchor=base west,
  text=blue,
  rotate=0.0
]{ k=0};
\draw[] (axis cs:2.3523833076031,2.10264095645153) -- (axis cs:2.3523833076031,2.10264095645153);
\node at (axis cs:-6.5,0.6)[
  scale=0.7,
  anchor=base west,
  text=red,
  rotate=0.0
]{ k=1};
\draw[] (axis cs:-8.19142736812755,4.18614376298692) -- (axis cs:-8.19142736812755,4.18614376298692);
\node at (axis cs:-15.2,2.9)[
  scale=0.7,
  anchor=base west,
  text=black,
  rotate=0.0
]{ k=2};
\node at (axis cs:18.5,22.25)[
  scale=1.5,
  anchor=base west,
  text=black,
  rotate=0.0
]{ \textsf{4}};
\end{axis}

\end{tikzpicture}
\caption{Illustration of Steps 1 to 4 of \textsf{CApR}-$n$. (Plot 1) Three depots (D0, D1 and D2) and 50 fields are visualised by the black and green balls, respectively. (Plot 2) The fields are assigned to $\tilde{k}=10$ clusters (C0,\dots,C9). All fields belonging to the same cluster are colored correspondingly. The cross-signs indicate the $\tilde{k}$ centroids and are labeled accordingly. (Plot 3)  Results of  \textsf{IP}-$7$ applied to the $\tilde{k}$ centroids. (Plot 4) Results of \textsf{CApR}-$7$. For visualisation, the fields are colored according to the clustering result.  Labels $k=0$, $k=1$ and $k=2$ indicate the first edge traversal of each crop-tour, whereby a \emph{crop-tour} denotes the harvesting routes associated with a specific crop $k$.}
\label{fig:alg_CApR}
\end{figure*}

\subsection{Analysis and Discussion of the Main Algorithm\label{subsec_AnalysisMainAlg}}

Several comments are made. First, clustering Step 2 is introduced to upscale the number of fields that can be handled for the coupling of crop assignment and routing. This is relevant if the available combination of computational power and IP-solver is prohibiting to solve a large-scale \textsf{IP}-$n$ with many fields. However, it is here explicitly emphasised that clustering induces suboptimality in the solution. This aspect is exemplatorily quantified in Experiment 2 of Sect. \ref{subsubsec_expt2}, and further discussed in the outlook of the conclusion in Sect. \ref{sec_conclusion}. Note that the clustering Step 2 does not necessarily have to be conducted according to spatial proximity of fields. Fields can be clustered arbitrarily. Also, single fields can be assigned to a single cluster for special analysis. Clustering Step 2 entails Step 4, which solves multiple instances of a \emph{TSP with different start and end vertex} and thus generates routes within clusters of fields planting the same crop. Therefore, omitting superscript $k$ for brevity, the general IP is
\begin{subequations}
\label{eq:TSPdifferentStartEnd}
\begin{align}
\min &\ \ \sum_{i<j}c_{ij} x_{ij} \label{eq_TSPdifferentStartEnd_objFcn}\\
\mathrm{s.t.} &\ \  \sum_{j=1}^{N-2} x_{0j}=1,~\sum_{j=1}^{N-2} x_{jN-1}=1,\label{eq_TSPdifferentStartEnd_1stcstrt}\\
& \ \  \sum_{i<l} x_{il} + \sum_{l<j} x_{lj}=2,~l=1,\dots,N-2,\label{eq_TSPdifferentStartEnd_2ndcstrt}\\
& \ \ \sum_{{i,j\in S}\atop{i<j}} x_{ij} \leq |S|-1,~\forall S\subseteq \mathcal{V}\backslash\{0,N-1\},~3\leq |S|\leq N-3,\label{eq_TSPdifferentStartEnd_SECs}\\
& \ \ x_{ij}\in\{0,1\},~0\leq i < j,~j=1,\dots,N-1, 
\end{align}
\end{subequations} 
where vertex-indices $0$ and $N-1$ are defined as the start and end vertex of the \emph{traveling salesman tour} (TSP) connecting the $N$ vertices of a given cluster, respectively. See Step 4 of \textsf{CApR}-$n$ for its application and the specific method for the selection of start and end vertices. See also Fig. \ref{fig:alg_CApR} for further visualisation. Constraints~\eqref{eq_TSPdifferentStartEnd_1stcstrt} and \eqref{eq_TSPdifferentStartEnd_2ndcstrt} indicate that start and end vertex are incident to one edge, and all other vertices incident to two, respectively. Under the assumption of symmetry, the \emph{subtour elimination constraints} (SECs) according to \cite{laporte1992traveling} are given by \eqref{eq_TSPdifferentStartEnd_SECs}.

Second and importantly, note that if the number of clusters is equal to the number of fields, i.e., $\tilde{k}=L$, the algorithm is reduced to the IP-formulations of Step 3. Then, \textsf{CApR}-$n$ is equivalent to \textsf{IP}-$n$, and $J^{\textsf{CApR}-n}=J^{\textsf{IP}-n}$. 
\begin{remark}
It is stressed that for $\tilde{k}<L$, because of the clustering Step 2, there can be examples constructed in which a hierarchical method that (a) first solves an assignment problem \emph{without} accounting for \emph{any} spatial field-proximity, before (b) then computing crop-tours can be more cost-efficient than a corresponding \textsf{CApR}-solution. However, this cannot occur for $\tilde{k}=L$. This is since the hierarchical solution is a feasible solution of the method coupling crop assignment and vehicle routing. Therefore, in practice, to optimise the monetary result the number of clusters must always be increased as much as computational power and available IP-solver permit, ideally, until $\tilde{k}=L$.
\label{remark_tildekSmallerL}
\end{remark} 
For completeness, a basic IP for pure crop assignment to fields  \emph{without accounting for any field-connectivity information} is $\text{min}\{-\sum_{l\in\mathcal{L}}\sum_{k\in\mathcal{K}}r_l^k\delta_l^k : \sum_{k\in\mathcal{K}} \delta_l^k=1,~\forall l\in\mathcal{L},~\delta_l^k\in\{0,1\},~\forall l\in\mathcal{L},~\forall k\in\mathcal{K}\}$. Under the additional assumption of field-uniform $r_l^k=r^k,~\forall l\in\mathcal{L}$, its optimal solution always assigns the most profitable crop indexed by $k^\star=\text{arg} \max_{k\in\mathcal{K}} r^k$ to \emph{all} fields. Furthermore, when including both crop rotation and diversification constraints, and introducing the relaxation $\sum_{k\in\mathcal{K}} \delta_l^k\leq 1$,  feasibility of the assignment IP is always guaranteed. This is since these constraints can always be satisfied by $\delta_l^k=0$. Moreover, note that  the solution of the LP-relaxation of the assignment IP, and also including \emph{crop rotation constraints}, is \emph{integer feasible}, and thus solves these problems as well. The proof is omitted for brevity, but it follows from the \emph{integrality} property of the LP-relaxation of the assignment IP, see  \cite{schrijver1998theory} and \cite{heller1956extension}. Consequently, very large instances (with many fields and crops) of the assignment OP can easily be solved. As a remark, the aforementioned inequality relaxation does not affect the integrality property. This is since slack variables $s_l$ can be introduced such that $\sum_{k\in\mathcal{K}} \delta_l^k + s_l = 1,~s_l\in \{0,1 \},~\forall l\in\mathcal{L}$. In contrast, by adding diversification constraints, in general, the LP-relaxation is rendered to not be integer feasible anymore.


Third, the relations between different \textsf{IP}-$n$ are discussed. 
\begin{proposition}\label{proposition_IP2vsIP3}
It always holds $J^{\textsf{IP}-3}\geq J^{\textsf{IP}-2}$.
\end{proposition} 
\begin{proof}
The proof is by contradiction. Assume $J^{\textsf{IP}-2}>J^{\textsf{IP}-3}$. $J^{\textsf{IP}-2}$ and $J^{\textsf{IP}-3}$ differ by cost coefficients $c_{dj}^k=c_{dj}^{k,k_\text{min}},~\forall k\in\mathcal{K}$ and $c_{dj}^k$ in \eqref{eq_def_cdjk}, respectively. By linearity of $J^{\textsf{IP}-2}$ and the definition of $c_{dj}^{k,k_\text{min}}$ according to \eqref{eq_def_cdjkmin}, and by nonnegativity of $\tilde{c}_{\tilde{d}j}^k$, $J^{\textsf{IP}-2}$ can always be increased by concentrating all harvesters, $ \sum_{\tilde{d}\in\mathcal{D}} N_{\tilde{d}}^{\text{harv},k}$, to the most cost-efficient depot. This is the \textsf{IP}-3 solution and therefore contradicts our assumption. The equality-part is because a special case of \textsf{IP}-2 is that none harvesters are initially located at any of the depots except the optimal one according to \textsf{IP}-3. 
\end{proof}
It also always holds that $J^{\textsf{IP}-3}\geq J^{\textsf{IP}-1}$. This is since the latter single depot case is a feasible solution of the former multiple depot case. Generalizing statements regarding $J^{\textsf{IP}-1}$ versus $J^{\textsf{IP}-2}$, and likewise for $J^{\textsf{IP}-3}$ vs. $J^{\textsf{IP}-4}$ cannot be made. This is because it is always possible to create counterexamples in favor of one or another solution. However, it always holds that $J^{\textsf{IP}-n}\geq J^{\textsf{IP}-(n+4)},~\forall n=1,\dots,4$. This is since the method of enforcing \emph{all} crops as for $\textsf{IP}-(n+4)$ is always a feasible solution of $\textsf{IP}-n$.


Fourth, the relations between various $J^{\textsf{IP}-n},~\forall n=1,\dots,8$, as discussed before, can in general \emph{not} be translated to the corresponding \textsf{CApR}-$n$ solutions for $\tilde{k}<L$. This is because of the clustering-heuristic in \textsf{CApR}-$n$. For instance, $J^{\textsf{CApR}-3}\geq J^{\textsf{CApR}-2}$ can in general not be guaranteed. 

Fifth, it is elaborated on Step 4 of \textsf{CApR}-$n$. Under the absence of priority constraints, there exist two directions in which to traverse any crop-tour. The traversal direction affects the closest fields between any pair of consecutive clusters. Consequently, the TSP-solution for each cluster, and thereby ultimately the total path length of the crop-tour, is affected, too. This motivated to test both cluster-sequences as indicated in Step 4. As stated, \textsf{CApR}-$n$ does not account for priority constraints, i.e., for a priori modeling of field ripeness sequences. Therefore, Step 2 and 4 require modification and clustering must be conducted according to an objective accounting for ripeness level. As a consequence, the traversal direction for Step 4 would also be fixed. 

Sixth, the result of \textsf{CApR}-$n$ could be further refined by including an additional  step (before Step 5) for heuristic local searches that iteratively tests and exchanges field-pairs within a crop-tour sequence if it improves the $J^{\textsf{CApR}-n}$-result. Alternatively, local field sequences could here also be exchanged manually according to heuristic preferences of the farm operator.

\section{Extensions\label{sec_Extensions}}

\subsection{Financial Considerations Regarding Leasing\label{subsec_leasing}}

For leasing considerations, the \emph{partial} service of a subset of fields is of interest. Let subset $\tilde{\mathcal{L}} \subseteq \mathcal{L}$ denote all fields for which we do not necessarily want to enforce field service but contemplate leasing options. Then, for \textsf{IP}-3, we maintain equality constraints \eqref{eq_ILP3_1stcstrt} and \eqref{eq_ILP3_2ndcstrt} only for $\mathcal{L}\backslash \tilde{\mathcal{L}}$, and define relaxed inequalities $\sum_{d\in\mathcal{D}} x_{dl}^k + \sum_{i<l} x_{il}^k + \sum_{l<j} x_{lj}^k \leq  2 \delta_l^k,\forall l\in\tilde{\mathcal{L}},\forall k\in\mathcal{K}$ and $\sum_{k\in\mathcal{K}} \delta_l^k \leq  1,~\forall l\in\tilde{\mathcal{L}}$. We similarly relax corresponding constraints for all other \textsf{IP}-$n$. Any \textsf{CApR}-$n$ including such constraints, shall be denoted as the relaxed \textsf{CApR}-$n$. In constrast, the original problem according to Sect. \ref{subsec_MainAlg} is referred to as the standard \textsf{CApR}-$n$.

An important financial consideration for every farm is to decide upon \emph{either} servicing \emph{or} renting out of one's fields, and additionally the decision upon taking of leases on additional fields for coverage. Let us denote the sets of corresponding fields by $\mathcal{L}^\text{own}$ (farmer's own fields), $\mathcal{L}^\text{pro}\subseteq \mathcal{L}^\text{own}$ (potential rent outs) and $\mathcal{L}^{ptl}$ (potential fields for taking leases upon), respectively. Then, the following algorithm provides guidelines for decision making:

\vspace{0.2cm}
\hspace{-0.4cm}
\begin{tabular}{ l l }
\hline\\[-8pt]
 \multicolumn{2}{l}{\textbf{Algorithm 2}: Renting out and Taking Leases}\\[3pt]
\hline\\[-8pt]
1: & Define all fields considered by $\mathcal{L}= \mathcal{L}^\text{own} \cup \mathcal{L}^\text{ptl}$.\\[3pt]
2: & Define the set of fields of interest by $\tilde{\mathcal{L}}= \mathcal{L}^\text{pro} \cup \mathcal{L}^\text{ptl}$. \\[3pt]
3: & Modeling according to farmer's production means.\\ 
& - determine parameters of Step 1 of \textsf{CApR}-$n$,$\forall l\in\mathcal{L}$.\\[3pt]
4: & \textbf{Solve} a relaxed \textsf{CApR}-$n$ for any desired $n=1,\dots,8$.\\[3pt]
5: & Determine $\mathcal{L}^\text{ntl} = \{l\in\mathcal{L}^\text{ptl}: \delta_l^k=0,\forall k\in\mathcal{K} \}$.\\
& - \emph{not} take a lease on any of these fields.\\[3pt]
6: & Determine $\mathcal{L}^\text{ro} = \{l\in\mathcal{L}^\text{pro}: \delta_l^k=0,\forall k\in\mathcal{K} \}$.\\
& - rent out these fields (\emph{any} positive return is good).\\[3pt]
7: & \textbf{Solve} standard \textsf{CApR}-$n$ for $\mathcal{L}^1 = \mathcal{L}\backslash \{ \mathcal{L}^\text{ntl} \cup \mathcal{L}^\text{ro} \}$.\\
& - denote its monetary return by $J_{\mathcal{L}^1}$.\\[3pt]
8: & \textbf{Solve} standard \textsf{CApR}-$n$ for $\mathcal{L}^\text{own}$.\\
& - denote its monetary return by $J^{\mathcal{L}^\text{own}}$.\\[3pt]
9: & Take leases of fields $\mathcal{L}^\text{ptl}\backslash \mathcal{L}^\text{ntl}$ for the overall payment\\
&   rate of at most $\Delta J = J^{\mathcal{L}^1}-J^{\mathcal{L}^\text{own}}$.\\[3pt]\hline\\[0pt]
\end{tabular}
It is elaborated on Step 4. Suppose the inclusion of a field does not improve the total financial return. This may be because of too expensive production costs, for example, when fields are located very distant apart from depots or due to limited harvesting windows. Then, renting out is profitable essentially already for any positive return. Let us also discuss Step 9. In contrast to pure assignment problems, the maximum leasing rate $\Delta J$ cannot easily be distributed among corresponding fields. This is because monetary profits are nonlinearly related to crop returns because of the coupling with routing decisions. Importantly, the precise distribution of leasing rates of individual fields is not relevant as long as it overall does not surpass $\Delta J$. Thus, $\Delta J$ provides the farmer with an upper bound on profitable leasing rates. If $\Delta J$ cannot be attained in negotiations, different $\mathcal{L}^\text{ptl}$ should be selected and the algorithm solved again. This is repeated until a corresponding upper bound can be satisfied, or, ultimately, $\mathcal{L}^\text{own}\backslash \mathcal{L}^\text{ro}$ are serviced.

The second financial consideration is motivated by the comparison of monetary returns for \textsf{CApR}-$n$. It permits to determine ``fair'' prices for leasing when sheltering machinery at the various depots. It is envisioned that all collaborating farmers first involve in accurate system modeling of cost coefficients, before then solving either all of  \textsf{CApR}-$n$, $\forall n=1,\dots,4$, or all of \textsf{CApR}-$n$, $\forall n=5,\dots,8$. Specifically, the difference in objective values between   
\textsf{CApR}-$2$ (or \textsf{CApR}-$7$ for enforcement of all $K$ crops in the solution) and the remaining \textsf{CApR}-$n$ then permits to determine an upper bound on leasing rates for depot usage.

Finally, eventhough a detailed discussion is here out of scope, the importance of a suitable method for sharing of profits is underlined. It is fundamental for providing incentives for farmers to collaborate and adopt proposed planning methods. See \cite{andersson2005farm} for a discussion about how farm cooperation can improve both sustainability and profitability.

\subsection{Application in Practice} 

Modeling of parameters listed in Step 1 of \textsf{CApR}-$n$ is largely affecting the overall result and may be based on historical field and crop yield data. By the selection of $c_{ij}^k$, computational complexity can be reduced by pruning specific undesired field connections from a path network, thereby implicitly also influencing priority constraints. Large fields often have multiple possible field entrance and exit points, which may significantly affect inter-field travel distances. Therefore, field coverage patterns and in-field navigation (see \cite{jensen2012field}, \cite{conesa2016mix}, \cite{plessen2018partial}, \cite{plessen2019optimal}) could also be co-planned to account for crop-tours efficiently linking fields that plant the same crops. This may be subject of future work.

\section{Numerical Simulations\label{sec_numerical_experiments}}

This section summarises the experimental data setup, the methods used to solve proposed IPs, and finally some numerical simulation results.

\subsection{Experiment Data Setup\label{subsect_ExptSetup}}

\begin{table}
\centering
\caption{Normalised and averaged monetary returns in Northern Germany that were used in simulations.} 
\label{tab_rk}
\begin{small}
\bgroup
\def\arraystretch{1}
\begin{tabular}{|c|c|c|c|c|}
\hline
\rowcolor[gray]{0.8} & Barley & Rapeseed & Wheat  & Unit \\[2pt]\hline 
\rowcolor[gray]{0.93} k & 0 & 1 & 2 & - \\[2pt]\hline 
$\tilde{r}^k$ & 570 & 600 & 750 & $\text{\euro}/\text{ha}$ \\\hline
\end{tabular}
\egroup
\end{small}
\end{table}

For numerical simulations one data set is composed of 10 different random instances. In Sect. \ref{subsec_NumExpts} the mean evaluation results are reported for these instances. Problem data is generated synthetically based on realistic parameter settings from farming in Northern Germany. The data generation for each instance is as follows. The number of available depots, crops and fields is set as $D=3$, $K=3$ and $L=50$, respectively. Field and depot locations are generated randomly according to a Gaussian distribution centered at the origin with standard deviations $\sigma_d=10$km, $\forall d\in\mathcal{D}$, and $\sigma_l=15$km, $\forall l\in\mathcal{L}$. To each depot, we randomly assign a number of harvesters according $N_d^{\text{harv},k}=\max(1,\left \lfloor{5u_d }\right \rfloor ),~ u_d\sim\mathcal{U}(0,1),~\forall d\in\mathcal{D}$, where $\mathcal{U}(0,1)$ denotes the Uniform distribution with zero mean and unit variance, and $\left \lfloor{\cdot}\right \rfloor$ denotes rounding to the next smallest integer. Normalised traveling costs per harvester and km are set uniformly (for all inter-depot and inter-field distances) as $\tilde{c} = 30 \frac{\text{\euro}}{\text{km}}$. A cost of $m=1000$\euro~for every planted crop is assumed. Maintenance costs are assumed to be identical for all depots. Therefore, w.l.o.g. $z^d=0,~\forall d\in\mathcal{D}$ is set. Realistic  normalised monetary returns in \euro~per ha and crop are determined as mean values from intermediate soil qualities and crop yields in Northern Germany. These are summarised in Table \ref{tab_rk}. Two options were considered for monetary return per field and crop. First, field sizes were generated according to $s_l = \max(20 + 10u_d,1),~u_d\sim\mathcal{N}(0,1),~\forall l\in\mathcal{L}$, where $\mathcal{N}(0,1)$ denotes the Gaussian distribution with zero mean and unit variance. In combination with $L=50$ this results approximately in a total coverage size of 1000ha. According to the survey by \cite{Landwirtschaft}, in all of Germany there are 299134 farming businesses of which only 1502 have a size of more than 1000ha. Field sizes were then multiplied with $\tilde{r}^k$ according to Table \ref{tab_rk}, to yield $r_l^k = s_l \tilde{r}^k,~\forall l\in\mathcal{L},~k\in\mathcal{K}$. This method of data generation is intuitive. However, since normalised monetary return is considerably higher for wheat in contrast to barley and rapeseed, the application of \textsf{CApR}-$n$ typically assigns wheat to all fields, unless crop rotation constraints, or diversification constraints, as in \textsf{CApR}-$n$ for $n=5,\dots,8$, are included. Thus, in the latter cases, the crop with smallest monetary return is assigned to the cluster with smallest field area, and the crop with second-smallest return to the second-smallest area and so forth. In a second setting, and to add more variety, monetary returns per field and crop were therefore generated differently according to $r_l^k = \max(20 + 10u_l^k,1) \tilde{r}^k,~\forall l\in\mathcal{L},~k\in\mathcal{K}$, with $u_l^k\sim\mathcal{N}(0,1)$. Throughout Sect. \ref{subsec_NumExpts}, the second setting is used for data generation. 

\subsection{Solution of Integer Programs\label{subsec_Soln_IPs}}

For the solution of IPs in \textsf{CApR}-$n$, three open-source IP-solver candidates were considered: CBC (\cite{forrestcbc}), GLPK\_MI (\cite{makhoringlpk}), and ECOS\_BB (\cite{Domahidi2013ecos}). All of these solvers were called through the domain-specific language CVXPY for optimisation embedded in Python (\cite{diamond2016cvxpy}). All numerical experiments were conducted on a laptop running Ubuntu 16.04 equipped with an Intel Core i7 CPU @2.80GHz$\times$8 and 15.6GB of memory. For the present applications and in preliminary tests (e.g., stochastic experiments for $L=50$ and various $\tilde{k}$), the GLPK\_MI consistently outperformed the other two considered IP-solvers. Therefore, it was used for all numerical experiments of Sect. \ref{subsec_NumExpts}. For clarification, the focus of this paper was on problem \emph{modeling} for a broad variety of settings typical in an agricultural context. The scope of this pape is not discussing or applying the most efficient (commercial) problem \emph{solver}. 
%

A remark to incorporating SECs in Integer Program is made.  The proposed optimisation problems include an exponential number of SECs (see also \cite{laporte1992traveling} for a general discussion). Therefore, in this paper SECs are approached in form of \emph{separation algorithms} according to \cite{pataki2003teaching}. Thus, SECs are added sequentially as they are needed. For instance, with regard of \eqref{eq:TSPdifferentStartEnd}, it is first solved without \eqref{eq_TSPdifferentStartEnd_SECs}. Then, if the result does not return any subtour, the optimal solution has been found. Otherwise, all detected subtours are added to \eqref{eq:TSPdifferentStartEnd} as SECs, and the IP is solved again. This is repeated until a solution without subtours is found (the optimal solution), or a maximal number of SEC-iterations is reached. 

\begin{table*}
\centering
\caption{Summary of Experiment 1. The results of \textsf{CApR}-$n$ for $n=1,\dots,8$ are compared for two different $\tilde{k}$. } 
\label{tab_expt_1}
\begin{small}
\def\arraystretch{1.15}
\begin{tabular}{|l|c|c|c|c|c|c|c|c|c|}
\hline
\rowcolor[gray]{0.8} \multicolumn{10}{|c|}{$\tilde{k}=10$}\\
\hline
\rowcolor[gray]{0.93} & Unit & $n=1$ & $n=2$ & $n=3$ & $n=4$ & $n=5$ & $n=6$ & $n=7$ & $n=8$  \\ \hline 
$\bar{J}^{\textsf{CApR}-n}$ & \euro  & 722075 & 718132 & 723664 & 715732 & 715442 & 711435 & 718474 & 715732 \\\hline
$\bar{T}_\text{CPU}^{\textsf{IP}-n}$ & s & 0.07 & 0.59 & 0.09 & 0.07 & 0.03 & 0.02 & 0.02 & 0.02 \\
$\bar{N}_\text{iterSEC+}^{\textsf{IP}-n}$ & - & 10.10 & 31.30 & 9.40 & 2.00 & 4.30 & 3.50 & 2.80 & 2.00 \\
$N_z^{\textsf{IP}-n}$ & - & 196 & 196 & 262 & 541 & 195 & 195 & 258 & 348\\
$N_\text{eq}^{\textsf{IP}-n}$ & - & 41 & 41 & 44 & 68 & 43 & 43 & 50 & 68\\
$N_\text{ineq,NoSEC}^{\textsf{IP}-n}$ & - & 197 & 197 & 275 & 1112 & 195 & 195 & 258 & 348\\
$\bar{N}_\text{ineq,finalIP}^{\textsf{IP}-n}$ & - & 208.70 & 232.6 & 285.6 & 1113.0 & 199.0 & 197.9 & 260.1 & 349.0\\\hline
$\bar{T}_\text{CPU}^\text{TSPs}$ & s & 0.03 & 0.03 & 0.03 & 0.03 & 0.91 & 0.82 & 0.03 & 0.03\\
$\bar{N}_\text{iterSEC+}^\text{TSPs}$ & - & 12.6 & 13.7 & 13.5 & 12.9 & 16.8 & 14.2 & 13.2 & 12.9\\\hline
$P_\text{conv}^{\textsf{IP}-n}$ & \% & 100 & 100 & 100 & 100 & 100 & 100 & 100 & 100 \\
\hline
\rowcolor[gray]{1} \multicolumn{10}{c}{}\\[-5pt]
\hline
\rowcolor[gray]{0.8} \multicolumn{10}{|c|}{$\tilde{k}=20$}\\
\hline
\rowcolor[gray]{0.93} & & $n=1$ & $n=2$ & $n=3$ & $n=4$ & $n=5$ & $n=6$ & $n=7$ & $n=8$  \\ \hline 
$\bar{J}^{\textsf{CApR}-n}$ & \euro & - & - & 748851 & 747230 & 746431 & 742329 & 748379 & 747230 \\\hline
$\bar{T}_\text{CPU}^{\textsf{IP}-n}$  &  s & 63.09 & 220.72 & 49.24 & 46.06 & 0.22 & 0.14 & 44.83 & 8.14 \\
$\bar{N}_\text{iterSEC+}^{\textsf{IP}-n}$  & - & 73.8 & 124.0 & 82.6 & 24.2 & 9.0 & 5.8 & 24.3 & 24.2\\
$N_z^{\textsf{IP}-n}$ & - & 691 & 691 & 817 & 1366 & 690 & 690 & 813 & 993\\
$N_\text{eq}^{\textsf{IP}-n}$ & - & 81 & 81 & 84 & 108 & 83 & 83 & 90 & 108\\
$N_\text{ineq,NoSEC}^{\textsf{IP}-n}$ & - & 692 & 692 & 830 & 2477 & 690 & 690 & 813 & 993\\
$\bar{N}_\text{ineq,finalIP}^{\textsf{IP}-n}$ & - & 815.8 & 918.0 & 958.1 & 2504.3 & 702.1 & 698.5 & 855.5 & 1020.3\\\hline
$\bar{T}_\text{CPU}^\text{TSPs}$ & s  & 0.02 & 0.02 & 0.02 & 0.02 & 0.02 & 0.02 & 0.02 & 0.02\\
$\bar{N}_\text{iterSEC+}^\text{TSPs}$ & - & 14.4 & 14.3 & 14.6 & 14.5 & 14.6 & 14.6 & 14.6 & 14.5\\\hline
$P_\text{conv}^{\textsf{IP}-n}$ & \%  & 90 & 30 & 100 & 100 & 100 & 100 & 100 & 100 \\ 
\hline
\end{tabular}
\end{small}
\end{table*}

The sequential inclusion of SECs as they are needed offers the advantage of simplicity. Furthermore and importantly, as will be shown empirically in Sect. \ref{subsec_NumExpts}, for the preferred IP-settings and $\tilde{k}=10$ \emph{very few} (low single digit) SEC-iterations were required to find an integer solution without subtours, which thus also implies efficiency. Nevertheless, it is here clarified that it may not be the most efficient method to handle SECs. This is further discussed in the outlook of Sect. \ref{sec_conclusion}, including suggestion of future work.

\subsection{Experimental Results\label{subsec_NumExpts}}

This section is partitioned into two parts. First, all proposed methods \textsf{CApR}-$n$ for $n=1,\dots,8$ are compared in experiments according to Sect. \ref{subsect_ExptSetup} for two different numbers of field-clusters. Second, the improvement potential of solving CApR-problems \emph{without} any clustering is illustrated, i.e., for $\tilde{k}=L$.

Evaluation criteria are averaged over 10 simulation experiments, and indicated by $\bar{y}$ for a criterion $y$. The percentage out of the 10 simulation experiments for which an \textsf{IP}-$n$ solution could be found in less than 200 SEC-iterations is denoted by $P_\text{conv}^{\textsf{IP}-n}$. The average number of decision variables, SEC-iterations plus the initial IP-iteration without SECs, number of equality constraints, inequality constraints when first omitting SECs and for the final SEC-iteration (before convergence) are denoted by $N_z^{\textsf{IP}-n}$, $\bar{N}_\text{iterSEC+}^{\textsf{IP}-n}$, $N_\text{eq}^{,\textsf{IP}-n}$, $N_\text{ineq,NoSEC}^{\textsf{IP}-n}$ and $\bar{N}_\text{ineq,finalIP}^{\textsf{IP}-n}$, respectively. Average accumulated CPU-time for the solution of all SEC-iterations and all TSP-problems are $\bar{T}_\text{CPU}^{\textsf{IP}-n}$ and $\bar{T}_\text{CPU}^{\text{TSPs}}$, respectively. The average number of required SEC-iterations for the solution of all TSPs (which includes all SEC-iterations plus the initial IP-solution without SECs) is denoted by $\bar{N}_\text{iterSEC+}^\text{TSPs}$.

\subsubsection{Experiment 1}

The results of Experiment 1 are summarised in Table \ref{tab_expt_1}. For $\tilde{k}=20$, $\bar{J}^{\textsf{CApR}-1}$ and $\bar{J}^{\textsf{CApR}-2}$ are not reported in the comparison since only 90\% and 30\% of experiments could be solved within 200 SEC-iterations. Several remarks can be made. First, it is noted how quickly computational complexity rises with increasing $\tilde{k}$. For example, $\bar{T}_\text{CPU}^{\textsf{IP}-1}=0.07$s for $\tilde{k}=10$, but $\bar{T}_\text{CPU}^{\textsf{IP}-1}=63.09$s for $\tilde{k}=20$. Second, \emph{fixing} the number of serviced crops as done for \textsf{CApR}-$n,~n=5,\dots,8$, notably reduces CPU-time; compare $\bar{T}_\text{CPU}^{\textsf{IP}-n}$ vs. $\bar{T}_\text{CPU}^{\textsf{IP}-(n+4)}$ for $n=1,\dots,4$ and $\tilde{k}=20$. Similarly, the average number of SEC-iterations $\bar{N}_\text{iterSEC+}^{\textsf{IP}-n}$ is affected. For every additional SEC-iteration, an additional IP with an increased number of SECs has to be solved. Third, for $\tilde{k}=20$ \textsf{CApR}-2 had difficulties in finding a solution within 200 SEC-iterations, see $P_\text{conv}^{\textsf{IP}-2}$. The intuitive explanation is that the optimal solution when optimizing over any subset of $K$ crops in combination with characteristic cost coefficient \eqref{eq_def_cdjkmin} is very sensitive to newly added SECs and accordingly quickly changes, which explains the many SEC-iterations, often exceeding the 200 SEC-iteration bound. Fourth, for $\tilde{k}=10$ on average only 1 SEC-iteration was required for \textsc{CApR}-8 (recall that $\bar{T}_\text{CPU}^{\textsf{IP}-n}$ and $\bar{N}_\text{iterSEC+}^{\textsf{IP}-n}$ count all IPs including the initial one \emph{without} any SECs). Thus, for (a) small $\tilde{k}$ (such as $\tilde{k}=10$) and (b) $n=5,\dots,8$ the method according to Sect. \ref{subsec_Soln_IPs} adding SECs as they are need appears appropriate. Fifth, profit $\bar{J}^{\textsf{CApR}-n}$ increases with $\tilde{k}$. This concept is further emphasised in the next Experiment 2.

\subsubsection{Experiment 2\label{subsubsec_expt2}}

The purpose of Experiment 2 is to illustrate the benefit on profit when solving for $\tilde{k}=L$, i.e., dismissing the heuristic clustering-step. For brevity, results are reported only for \textsf{CApR}-5. However, they are by trend comparable for all other 7 methods. In contrast to Experiment 1 with $L=50$, the results for IP-5 are summarised in Table \ref{tab_expt_2} for $L=40$ (due to an exploding computational complexity). For $\tilde{k}=40$ there are 2580 integer variables and 2580 inequality constraints even without incorporation of any SECs. Note that for the given experiment with the maximum possible number of clusters $\tilde{k}=L$ a notable improvement of 15.4\% over the solution for $\tilde{k}=10$ could be observed. This underlines the economic incentive for attempting to solve IPs for $\tilde{k}=L$ as stressed in Remark \ref{remark_tildekSmallerL}, if this is permitted by the available combination of computational hardware and IP-solver.

\begin{table}
\centering
\caption{Summary of Experiment 2 for \textsf{CApR}-5. The improvement potential of solving \textsf{CApR}-5 \emph{without} any clustering step for $\tilde{k}=L$ is illustrated.} 
\label{tab_expt_2}
\begin{small}
\bgroup
\def\arraystretch{1.15}
\begin{tabular}{|l|c|c|c|c|}
\hline
\rowcolor[gray]{0.8} \multicolumn{5}{|c|}{\textsf{CApR}-5}\\
\hline
\rowcolor[gray]{0.93} & Unit &  $\tilde{k}=10$ & $\tilde{k}=40$ & $\Delta_\text{rel}$ [\%] \\ \hline 
$\bar{J}^{\textsf{CApR}-n}$ & \euro & 590911 & 681900 & \textbf{15.4}  \\\hline
$\bar{T}_\text{CPU}^{\textsf{IP}-n}$ & s  & 0.01 & 160.1 & 1.6e6 \\
$\bar{N}_\text{iterSEC+}^{\textsf{IP}-n}$ &- & 1.9 & 39.9 & 2000  \\
$N_z^{\textsf{IP}-n}$ &-& 195 & 2580 & 1223 \\
$N_\text{eq}^{\textsf{IP}-n}$ &-& 43 & 163 & 279\\
$N_\text{ineq,NoSEC}^{\textsf{IP}-n}$ &-& 195 & 2580 & 1223 \\
$\bar{N}_\text{ineq,finalIP}^{\textsf{IP}-n}$ &-& 196.2 & 2650.0 & 1251\\\hline
$\bar{T}_\text{CPU}^{TSPs}$ & s & 0.03 & 0 & -100 \\
$\bar{N}_\text{iterSEC+}^{TSPs}$ &-& 12.20 & 0 & -100  \\\hline
$P_\text{conv}^{\textsf{IP}-n}$ & \% & 100 & 100 & 0 \\ 
\hline
\end{tabular}
\egroup
\end{small}
\end{table}

\section{Conclusion\label{sec_conclusion}}

\vspace{-0.1cm}

A flexible framework for the coupling of crop assignment with vehicle routing for harvest planning was presented. This problem is relevant since the decision about crop assignment must be addressed by every farm manager at the beginning of every work-cycle starting with plant seeding and ending with harvesting. The main contribution was the proposal of 8 different IP formulations. It was found in numerical experiments that the 4 cases with enforced inclusion of any crop out of a set of crops to be computationally notable more efficient. This enforcement is applicable in practice since the list of eligible crops typically is very limited. For large-scale applications where sole IP formulations are not tractable anymore, a heuristic algorithm was proposed combining the IP-formulations with clustering of fields and the solution of local TSPs.

The focus of this paper was on problem \emph{modeling} for a broad variety of settings typical in agriculture. For future work the main task is development of an efficient \emph{solution method}. This is motivated by the fact that the output of \textsf{CApR}-$n$ remains heuristic and in general suboptimal for all $\tilde{k}<L$ due to the clustering step. Thus, IP-solution methods must be developed that enable to solve in reasonable time for $\tilde{k}=L$ and for large $L$, e.g., $L\approx 100$. Therefore, two main approaches are envisioned. First, by (a) focusing on one specific IP-formulation (e.g., \textsf{IP}-5 since \textsf{IP}-6 can be reduced to it and \textsf{IP}-7 can be solved by solving \textsf{IP}-5 for all available depots) the problem formulation can be standardised, before (b) a customised IP-solver, e.g., a branch-and-cut algorithm may  be developed. Reformulations of the SECs, for example, in form of MTZ-SECs (\cite{miller1960integer}), which introduce additional continuous variables for SECs and thereby render the problem of mixed integer nature may also be explored. Alternatively and secondly, \emph{optimisation by simulation} algorithms, e.g., Tabu search heuristics similarly to \cite{gendreau1994tabu}, may be tested to solve the proposed IP-problems.

\vspace{-0.3cm}

\bibliographystyle{model5-names}
\bibliography{mybibfile.bib}
\nocite{*}







\end{document}